\RequirePackage[l2tabu,orthodox]{nag}
\documentclass[11pt]{article}

\usepackage{authblk}
\usepackage{xparse,zref-savepos}
\makeatletter
\@ifundefined{zsaveposx}{\let\zsaveposx\zsavepos}{}
\newcounter{hposcnt}
\renewcommand*{\thehposcnt}{hpos\number\value{hposcnt}}
\NewDocumentCommand{\lplabel}{o m}{%
  \stepcounter{hposcnt}%
  \zsaveposx{\thehposcnt l}%
  \zref@refused{\thehposcnt l}%
  \zref@refused{hpos0l}%
  \makebox[0pt][r]{\makebox[\dimexpr\zposx{\thehposcnt l}sp-\zposx{hpos0l}sp][l]{#2}}%
  \IfNoValueF{#1}
    {\def\@currentlabel{#2}\ltx@label{#1}}
}
\makeatother
\AtBeginDocument{\zsaveposx{hpos0l}}

\usepackage{subcaption}

\usepackage{caption} 
\usepackage[ruled,vlined]{algorithm2e}

\usepackage{amsmath, amssymb, amsthm,dsfont,mathrsfs} 
\usepackage{dsfont}
\usepackage[all,warning]{onlyamsmath}
\usepackage{graphics, graphicx} 
\usepackage{array}

 \usepackage{setspace}

\usepackage[font=small,skip=0pt]{caption}

\usepackage[margin=1in]{geometry} 
\usepackage{booktabs} 
\usepackage{units} 

\usepackage{microtype} 

\usepackage{fancyhdr} 
\usepackage{url} 
\usepackage[normalem]{ulem}
\usepackage[dvipsnames]{xcolor} 
\usepackage{multicol}

\usepackage{diagbox}

\usepackage{tikz} 
\usetikzlibrary{shapes,decorations}
\usetikzlibrary{arrows}
\usepackage{pgfplots} 
\pgfplotsset{compat=newest} 

\usepackage{hyperref} 

\newcommand{\prob}{\ensuremath{\mathbb{P}}}
\newcommand{\expec}{\ensuremath{\mathbb{E}}}
\newcommand{\ind}{\ensuremath{\mathbf{I}}}
\newcommand{\reals}{\ensuremath{\mathbb{R}}}

\newcommand{\defeq}{\ensuremath{\triangleq}}

\newcommand{\sdleq}{\ensuremath{\preccurlyeq_{\mathsf{sd}}}}
\newcommand{\setS}{\ensuremath{\mathcal{S}}}

\newcommand{\argmax}{\ensuremath{\arg\max}}

\newcommand{\imix}{\ensuremath{\mathcal{I}_{\text{mix}}}}
\newcommand{\izero}{\ensuremath{\mathcal{I}_{0}}}
\newcommand{\ione}{\ensuremath{\mathcal{I}_{1}}}

\newcommand{\uihat}{\ensuremath{\overline{i}}}
\newcommand{\lihat}{\ensuremath{\underline{i}}}

\newtheorem{theorem}{Theorem}
\newtheorem{proposition}{Proposition}

\newtheorem{lemma}{Lemma}

\newtheorem{definition}{Definition}
\newtheorem{assumption}{Assumption}

\title{Information Design in Spatial Resource Competition}
\date{}
\author[1]{\small Pu Yang}
\author[2]{\small Krishnamurthy Iyer}
\author[1]{\small Peter Frazier}
\affil[1]{\small Cornell University}
\affil[2]{\small University of Minnesota}

\begin{document}
\maketitle

\begin{abstract}

  We consider the information design problem in spatial resource
  competition settings. Agents gather at a location deciding whether
  to move to another location for possibly higher level of resources,
  and the utility each agent gets by moving to the other location
  decreases as more agents move there. The agents do not observe the
  resource level at the other location while a principal does and the
  principal would like to carefully release this information to
  attract a proper number of agents to move. We adopt the Bayesian
  persuasion framework and analyze the principal's optimal signaling
  mechanism design problem. We study both private and public signaling
  mechanisms. For private signaling, we show the optimal mechanism can
  be computed in polynomial time with respect to the number of
  agents. Obtaining the optimal private mechanism involves two steps:
  first, solve a linear program to get the marginal probability each
  agent should be recommended to move; second, sample the moving
  agents satisfying the marginal probabilities with a sequential
  sampling procedure.  For public signaling, we show the sender
  preferred equilibrium has a simple threshold structure and the
  optimal public mechanism with respect to the sender preferred
  equilibrium can be computed in polynomial time. We support our
  analytical results with numerical computations that show the optimal
  private and public signaling mechanisms achieve substantially higher
  social welfare compared with no information or full information
  benchmarks in many settings.

\end{abstract}

\section{Introduction}

We study information design in the spatial resource competition settings, where a group of agents migrate across a network of locations, competing for stochastic time-varying resources at each location. This setting characterizes many real world scenarios: on crowd-sourced transportation platforms, drivers migrate across different neighborhoods of a city, competing for ride demands; in unskilled labor markets, workers migrate across different cities for more job opportunities; in nomadic animal husbandry, pastoralists migrate across different range lands, competing for water and pastures, etc.

In these scenarios, information about other locations largely affects an agent's decision about whether to leave her current location, and where to explore. However, such information is usually limited or even unavailable to the agents. Meanwhile, in many cases, there is a principal who has access to more information than any individual agent, e.g., platforms such as Uber and Lyft in the ride sharing market, or the government agencies in labor markets, or non-profits in the case of pastoralists. Using this information, the principal can influence the decisions of the migrating agents, better locate them and improve the total social welfare. For example, Uber and Lyft show the demand trend at different neighborhoods as a heatmap to the drivers \cite{ubersurge} in order to locate them to areas with higher demands, and government agencies provide information about employment and job openings in various sectors.  

However, sharing information may not always bring higher welfare, if the information is not released in a careful way: if a sports match or a concert at a stadium is about to finish and Uber informs every available driver in the nearby area, too many drivers may flood to this area, with many of them failing to find a customer due to over-supply. In such cases, it might be wiser for the platform to only inform a particular group of drivers to better match supply and demand. How to choose such a group? More generally, spatial resource competition scenarios typically exhibit negative externalities. In presence of such negative externalities, how should a principal effectively communicate her information to the agents in order to better position them?

Answering this question for a large network of locations is challenging: the effect of a signal that attracts a group of agents to a particular location may percolate across the entire network through agents' subsequent migration. On the other hand, signals can take very complicated forms in large networks, depending on many factors, including the number of locations and agents, the dynamics at each location, and the agents' belief about the state of the system and other agents' strategies. With the goal of obtaining insights to this spatial information design problem while retaining a tractable analysis, we consider a two-location model, which serves as a foundation for more complicated analysis.

\subsection{Overview of Model and Main Results}

We consider a model with 2 locations, $\ell_0$, $\ell_1$, and $N$ agents. Initially, all the agents are at $\ell_0$, and must decide whether to move to $\ell_1$, which has a stochastic resource. The utility each agent receives upon moving to $\ell_1$ depends on the stochastic resource level at $\ell_1$, as well as how many other agents move there. The agents do not know the resource level at $\ell_1$ while a principal can observe it. The principal would like to design a mechanism to share this information to the agents in order to attract a proper number of agents to move to $\ell_1$. 

We adopt the framework of Bayesian persuasion \cite{kamenica2011bayesian, rayo2010optimal} to study this information design problem. The principal's goal is to choose a signaling mechanism that maximizes the expected social welfare, i.e., the total expected utility of the agents.  In this work, we consider both private signaling mechanisms, where the principal sends personalized signals to each agent privately, as well as public signaling mechanisms, where the principal sends the same information to all the agents.

The standard approach using a revelation-principle style argument \cite{bergemann2019information} to find the optimal private signaling mechanism leads to a linear program in $2^N$ variables, rendering the computation challenging. Instead, our first main result in Section~\ref{sec:private-signal} characterizes a computationally efficient two-step approach to find the optimal private signaling mechanism. First, we perform a change-of-variables and instead of the signaling mechanism, we focus on the marginal probabilities $p_{ik}$ that an agent $i$ is recommended to move to $\ell_1$ along with $k-1$ other agents, for each $i$ and $k$. We show that the marginal probabilities $\{ p_{ik}\}_{i,k}$ corresponding to the optimal private signaling mechanism can be found by solving a linear program in $O(N^2)$ variables. Then, we describe an efficient sampling procedure that samples sets of agents according to the the optimal marginal probabilities $\{p_{ik}\}_{i,k}$. The optimal private signaling mechanism then asks the sampled set of agents to move to $\ell_1$ and the rest to stay at $\ell_0$. Finally, we provide a condition on the model parameters under which the optimal signaling mechanism has a simple threshold structure and can be computed in $O(\log N)$ time. 


Although private signaling mechanisms provide the principal more flexibility, a number of practical concerns often render private mechanisms infeasible \cite{dughmi2017algorithmic, lingenbrink2018retail, candogan2017optimal}. First, private mechanisms make the strong assumption of no "information leakage" among the agents, i.e., the agents do not share their personalized information with each other. This assumption may easily fail in practice. Furthermore, fairness considerations may prevent a principal from sharing different information with different agents; a fair-minded principal may even seek to avoid the semblance of providing conflicting information to different agents. 

Owing to these reasons, in Section~\ref{sec:public-signal}, we analyze the problem of finding the optimal public signaling mechanism, where the principal shares the same information with all the agents. To do this, we first characterize the equilibria of the incomplete information game among the agents subsequent to receiving any public signal. While the equilibrium set is quite complex, we show that for any common posterior belief of the agents, the equilibrium that maximizes the social welfare has a simple threshold structure. Using this result, we show that the optimal public signaling mechanism can be found as a solution to a linear program with $O(N)$ variables and constraints. Furthermore, we show that the optimal mechanism only randomizes over two signals.


Finally, we numerically investigate the performance of the optimal private and public signaling mechanisms, and show that these achieve substantially higher social welfare than the no-information or full-information mechanisms. 

The main point of departure of our work from past literature on information design is the modeling of negative externality among the agents. Past work on information design has focused mainly on settings with no externalities \cite{dughmi2017algorithmic} or settings where there is positive externalities among the agents \cite{candogan2017optimal, candogan2019persuasion}. In such settings, correlation among the agents' choice of action is beneficial, whereas the main challenge in our work is in de-correlating the agents' actions. This is especially challenging under public signaling, where public signals naturally tend to correlate the agents' actions. 

To conclude this section, we note that our approach for finding the optimal private signaling mechanism is not restricted to our model, but applies more broadly to settings where an agents' utility upon taking the action depends on the number of agents who take that action. Thus, our results on the computation of the optimal private signaling mechanism might be of independent interest to the research community.

\subsection{Related Work}

Our work focuses on information design in a resource competition scenario, adopting the Bayesian persuasion framework. We briefly survey here the related works on information design and Bayesian persuasion.

Information design problems on how a sender should persuade one or more receivers date back to \cite{grossman1981informational} and \cite{milgrom1981good}. The two mainstream frameworks studying this problem is the Bayesian persuasion framework, originating from \cite{kamenica2011bayesian, rayo2010optimal} and the ``cheap talk'' model \cite{crawford1982strategic,farrell1987cheap}, where the main difference is the former assumes the signal sender has the power to commit to a particular information sharing mechanism. 

\cite{kamenica2011bayesian} studies the basic setting with one sender and one receiver. \cite{taneva2015information,bergemann2016bayes,mathevet2017information} considers the more general ``one sender, many receivers'' setting and provide unified frameworks of deriving the optimal mechanism in such general settings. \cite{gentzkow2017bayesian} considers a setting where multiple senders wish to influence one receiver. 

For general ``one sender, many receivers'' information sharing settings, tractable computation of the optimal mechanism remains an open problem. A lot of works have studied various special cases. \cite{babichenko2017algorithmic, arieli2019private, dughmi2017algorithmic} considered the simplest scenario: each receiver's action imposes no externalities on other agents. \cite{babichenko2016computational, babichenko2017algorithmic, arieli2019private} characterize polynomial time computable optimal mechanisms when the sender's utility is supermodular or anonymous submodular. \cite{dughmi2017algorithmic} provide an $(1-1/e)$ optimal mechanism for general submodular sender utilities. For more general settings where agents' payoffs depend on other agents' actions,
 \cite{arieli2019private}
provides polynomial time computable optimal mechanism for binary action settings when the sender's utility is supermodular, and similar result is given in \cite{candogan2017optimal, candogan2019persuasion}. Both works point out that the optimal policy correlates recommendations to take the positive action as much as possible. Our work considers resource competition in settings where agents' actions have negative externalities, and to the best of our knowledge, is the first to study such settings. 

Information design and Bayesian persuasion are studied in many other settings, including voting \cite{wang2013bayesian,alonso2016persuading,bardhi2018modes}, ad auction \cite{badanidiyuru2018targeting}, 
online retailing \cite{lingenbrink2018retail, drakopoulos2018persuading}, bilateral trade \cite{bergemann2015limits}, advertising \cite{chakraborty2014persuasive}, security games \cite{xu2015exploring}, customer queues with delays \cite{lingenbrink2017queue}, Stackelberg competition between firms \cite{xu2016signaling} and team formation \cite{hssaine2018information}, etc. A more thorough  review of this topic can be found in \cite{dughmi2017survey}.
Beyond the one-time persuasion setting considered in most of the work we listed, several other works 
\cite{caillaud2007consensus,kremer2014implementing} considered the problem of sequentially persuading a group of agents.

\section{Model and Preliminaries}
\label{sec:model-signal}

\subsection{Model}

We consider a model with $N$ agents, a principal, and $2$ locations,
denoted by $\ell_0$ and $\ell_1$. Initially, all the agents are at
location $\ell_0$. There is a stochastic resource at $\ell_1$, with
the resource level denoted by $\theta$. We primarily focus on the
binary setting, with $\theta \in \Theta \defeq \{0, 1\}$ capturing the
presence or absence of a resource; our model can be extended to the
case where $\Theta$ is a finite set. Without observing $\theta$, each
agent at $\ell_0$ independently decides whether or not to move to
$\ell_1$, where she obtains a utility that depends on $\theta$, as
well as the number of other agents who also choose to move to
$\ell_1$. In addition, we assume that each agent incurs a moving cost
if she moves to $\ell_1$.

Formally, each agent $i\in [N]$ simultaneously chooses an action
$a_i \in \{0, 1\}$, where $a_i=0$ implies the agent chooses to stay at
$\ell_0$ and $a_i = 1$ implies she chooses to move to $\ell_1$. Let
$a = (a_i, a_{-i})$ denote the profile of actions chosen by all the
agents, and $A \defeq \{0,1\}^N$ denote the set of action
profiles. Note that for any $a \in A$, the number of agents that
choose to move to $\ell_1$ is given by $n(a) = \sum_{i=1}^N
a_i$. Then, for any action profile $a$ and any resource level
$\theta$, an agent $i$'s utility, $U_i : \Theta \times A \to \reals$,
is given by
\begin{align*}
  U_i(\theta, a) =
  \begin{cases}
    \theta \cdot F\left(\sum_{j=1}^N a_j \right) - r(i), & \text{if $a_i=1$;}\\
    0, & \text{if $a_i = 0$.}
  \end{cases}
\end{align*}
Here, $F: [N] \to \reals_+$ is the resource sharing function that
determines how the resource at $\ell_1$ is shared among the agents at
$\ell_1$, and $r(i)$ denotes agent $i$'s moving cost. In particular,
an agent $i$ who chooses to move to $\ell_1$ receives an utility of
$\theta F(n(a))$ from the resource, and incurs a cost $r(i)$ for
moving. Furthermore, we have normalized the utility of staying at
$\ell_0$ to be zero. Without loss of generality, we assume that $r(i)$
is increasing in $i$, i.e., $r(1) \leq r(2) \leq \cdots r(N)$. For
notational convenience, in the following, we let $r(0) = 0$ and
$F(0) = F(1)$.


We make the following assumptions on the resource sharing function
$F$:
\begin{assumption} The resource sharing function $F(n)$ is decreasing
  and convex in $n \geq 1$.  Furthermore, the total utility from the
  resource $nF(n)$ is increasing and concave in $n \geq 0$.
\end{assumption}
The first condition is meant to capture the fact that in most resource
sharing settings, the amount of resource each agent receives decreases
as competition increases, with the decrement diminishing with the
level of competition due to market saturation. On the other hand, the
second condition captures the fact that as the competition increases,
the total level of resource available to the agents also increases,
albeit also at a diminishing rate. This is especially common
in platform markets, where the presence of more agents on one side leads to
better service quality for the other side of the market, leading to
more conversion.


\subsection{Information structure}

We assume that the principal and the agents hold a common prior belief
$\mu \in \Delta(\Theta)$ about the resource level $\theta$, where
$\mu(1)$ denotes the prior probability that $\theta=1$, and
$\mu(0) = 1 - \mu(1)$. While the resource level $\theta$ is unobserved
by the agents, we assume that the principal observes $\theta$ prior to
the agents' choice of actions. The principal's goal is to communicate
this information about $\theta$ to the agents prior to their moving
decision, in order to better position them at the two locations. (We
describe the principal's objectives in more detail below.)

Following the methodology of Bayesian persuasion, the principal
commits to a {\em signaling mechanism} as a means to share information
with the agents. Formally, a signaling mechanism $(\Sigma, \phi)$
consists of a signal set $\Sigma$ and a signaling scheme
$\phi: \Theta \to \Delta(\Sigma^N)$. Given a signaling mechanism
$(\Sigma, \phi)$ and upon observing the resource level $\theta$, the
principal first chooses a signal profile
$s = (s_1, \cdots, s_N) \in \Sigma^N$ with probability
$\phi(s|\theta) \in [0,1]$.  Then, the principal (privately) reveals
the signal $s_i$ to agent $i$. In particular, $s_i$ is not revealed to
an agent $j \neq i$.

Note that $\phi(s|\theta)$ denotes the conditional probability that
the signal profile is $s$, given the resource level is $\theta$. In
particular, we have $\sum_{s \in \Sigma^N} \phi(s|\theta) = 1$ for
each $\theta \in \Theta$.  Analogously, we define
$\phi(\theta, s) = \mu(\theta) \phi(s | \theta)$ to be the unconditional
joint probability that the signal profile is $s$ and the resource
level is $\theta$. Finally, we let $\phi(s)$ denote the probability of
the principal selecting the signal profile $s \in \Sigma^N$, given by
$\phi(s) = \phi(0,s) + \phi(1,s)$.

A special case of a signaling mechanism is a {\em public} signaling
mechanism, where the principal publicly announces the information
about $\theta$ to all the agents. In other words, the principal always
shares the same information with all the agents. Such public signaling
can be captured by a signaling mechanism $(\Sigma, \phi)$ where
$\phi(s|\theta) =0$ for any $s \in \Sigma^N$ with $s_i \neq s_j$ for
some $i,j \in [N]$. Finally, we refer to any signaling mechanism that
is not public as a private signaling mechanism.



\subsection{Strategies and equilibrium}
\label{sec:actions-equilibrium}

Since an agent $i$ does not have access to the signals of the other
agents, she maintains a belief over both the resource level $\theta$
and the signals $s_{-i}$ of the other agents. Upon receiving her
signal $s_i$ from the principal, the agent updates her belief using
Bayes' rule, before deciding whether to move.  Let $q_i(\cdot | s_i)$
denote posterior belief of agent $i$ about the resource level and the
signals sent to other agents, given by
\begin{align}
\label{eq:bayes-update}
  q_i(\theta, s_{-i} | s_i) = \frac{\mu(\theta)\phi(s_i, s_{-i} | \theta)}{\sum_{\theta', s'_{-i}} \mu(\theta')\phi(s_i, s'_{-i} | \theta')}.
\end{align}
A strategy $p_i : \Sigma \to [0,1]$ for agent $i$ specifies, for each
possible signal $s_i$, the probability $p_i(s_i)$ with which she
decides to move to location $\ell_1$. Each agent, given her posterior
belief and the strategies $p_{-i}$ of the other agents, seeks to
choose a strategy $p_i$ that maximizes her expected utility. More
precisely, given a signaling mechanism $(\Sigma, \phi)$ and the
strategies $p_{-i}$ of the other agents, the expected utility of an
agent $i$ for moving to location $\ell_1$ upon receiving a signal
$s_i$ is given by
\begin{align*}
  u_i(s_i, \mathsf{move},  p_{-i}) = \expec_{q_i} [ U_i(\theta,a_i=1, a_{-i}) | a_{-i} \sim  p_{-i}(s_{-i})],
\end{align*}
Then, in a {\em Bayes-Nash equilibrium}, each agent $i$, upon
receiving a signal $s_i$, decides to move to $\ell_1$ if her expected
utility $u_i(s_i, \mathsf{move}, p_{-i})$ is positive. We have the
following formal definition:
\begin{definition}
  Given a signaling mechanism $(\Sigma, \phi)$, a strategy profile
  $(p_1, \ldots, p_N)$ forms a Bayes-Nash equilibrium (BNE), if for
  each $i \in [N]$ and $s_i \in \Sigma$,
  \begin{align*}
    p_i(s_i) &= \begin{cases} 1 & \text{if $u_i(s_i, \mathsf{move},
        p_{-i}) > 0$;}\\
      0 & \text{if $u_i(s_i, \mathsf{move},
        p_{-i}) < 0$.}
    \end{cases}
  \end{align*}
  (If $u_i(s_i,\mathsf{move}, p_{-i}) = 0$, then
  $p_i(s_i) \in [0,1]$.)
\end{definition}


As mentioned earlier, the principal's goal is to choose a signaling
mechanism to maximize the {\em expected social welfare}. Formally, the
social welfare $W : \Theta \times A \to \reals$ is defined as
\begin{align*}
  W(\theta, a) \defeq \theta \cdot n(a) \cdot  F(n(a))  - \sum_{j=1}^N a_j r(j).
\end{align*}
Here the first term denotes the total utility obtained by the
$n(a) = \sum_{j=1}^N a_j$ agents from the resource at $\ell_1$, and
the second term denotes the total moving costs incurred by the
agents. We assume that the principal knows the agents' moving
costs. Then, the principal's decision problem is to choose a signaling
mechanism $(\Sigma, \phi)$ such that in a resulting Bayes-Nash
equilibrium $(p_1, \cdots, p_N)$ among the agents, the expected social
welfare, given by
$\expec[ W(\theta, a) | a_i \sim p_i(s_i) , (\theta, s) \sim \phi]$,
is maximized. In the next section, we study the problem of computation
of the optimal signaling mechanism and characterize its
structure. Subsequently, in Section~\ref{sec:public-signal}, we
analyze the related problem of optimal public signaling, where the
principal is restricted to sharing information via public signaling
mechanisms.


\section{Private Signaling Mechanism}
\label{sec:private-signal}


From a standard revelation-principle style argument
\cite{kamenica2011bayesian,bergemann2019information}, there exists a
{\em straightforward} and {\em persuasive} signaling mechanism that optimizes the expected social welfare. In a straightforward mechanism, the principal makes an action recommendation to each agent, and if it is optimal for each agent to follow the recommendation (assuming all others do so), then the mechanism is said to be persuasive. Thus, to obtain an optimal private mechanism, it is sufficient to restrict our attention to persuasive straightforward mechanisms. This implies that it suffices for the principal to determine the subset of agents to recommend to move for each resource level. 

For $\theta=0,1$ and $S \subseteq [N]$, let $\phi(S |\theta)$ be the probability of recommending the set $S$ of agents to move to $\ell_1$,  given resource level $\theta$. We abuse the notation to let
 $$W(\theta, S) \defeq \theta |S|F(|S|) - \sum_{i \in S} r(i)$$
be the social welfare when resource level is $\theta$ and agents in $S$ move to $\ell_1$. The optimal signaling scheme $\phi$ is then obtained as solution to the following linear program:
\begin{align}
\max_{\phi} \quad &\sum_{\theta=0,1} \mu(\theta) \sum_{ S \subseteq [N]} \phi(S |\theta) W(\theta, S) \label{eq:old-objective} \\
s.t. \quad & \sum_{\theta=0,1} \mu(\theta) \sum_{ S: i\in S} \phi(S | \theta)(\theta F(|S|) - r(i)) \geq  0, \quad i \in [N],  \label{eq:move-persuasive} \\
\lplabel[eq:lp-private]{(LP.1)} & \sum_{\theta=0,1} \mu(\theta) \sum_{S: i\notin S} \phi(S | \theta)(\theta F(|S|+1) - r(i)) \leq  0, \quad i \in [N],  \label{eq:stay-persuasive} \\
& \sum_{S \subseteq [N]} \phi(S | \theta) = 1, \quad \theta=0,1,  \notag \\
& \phi(S | \theta) \geq 0, \quad  \theta=0,1;  S \subseteq [N]. \notag
\end{align}
The first two sets of constraints ensure $\phi$ is persuasive: the first constraint states that any agent $i$ who is recommended to move to $\ell_1$ must have non-negative utility for moving, whereas the second constraint states that any agent $i$ who is recommended to stay must have non-positive utility for moving. The other constraints ensure that $\phi(\cdot | \theta)$ is a valid probability distribution for both $\theta=0,1$. Note that this linear program has $O(2^N)$ variables, and is computationally challenging. As a first step towards simplifying the problem, we note the following lemma, which allows us to only consider mechanisms that recommend all agent to stay at $\ell_0$ when $\theta=0$.
\begin{lemma}
\label{lem:no-move-zero-theta}
For any persuasive mechanism, there exists another persuasive mechanism that recommends every agent to stay at $\ell_0$ when $\theta=0$, and achieves a higher social welfare than the original mechanism.
\end{lemma}

We provide proofs of all results in Appendix~\ref{ap:proofs}. Although Lemma \ref{lem:no-move-zero-theta} reduces the size of \ref{eq:lp-private} by half, this linear program is still computationally challenging. However, taking a closer look at each agent's utility function, we notice that each agent $i$'s payoff for moving when $\theta=1$ depends only on how many other agents are moving. With this observation, we now consider an alternative formulation of \ref{eq:lp-private}.

Given a persuasive signaling scheme $\phi$ satisfying Lemma~\ref{lem:no-move-zero-theta}, i.e., $\phi(\varnothing | 0) = 1$, we define $p_{ik}$ to be the joint probability,  given $\theta=1$, that this signaling scheme recommends $k$ agents to move and agent $i$ is among them, i.e., 
\begin{equation}
\label{eq:pik-def}
p_{ik} = \sum_{S} \phi(S | 1) \cdot \ind\{ i \in S, |S| = k\}.
\end{equation}
The following lemma allows us to write the objective and persuasive constraints of \ref{eq:lp-private} in terms of $p=\{p_{ik}: i \in [N], k \in [N]\}$.
\begin{lemma}
\label{lem:reform-lp}
The objective \eqref{eq:old-objective} of \ref{eq:lp-private} can be written as
\begin{equation*}
\mu(1) \sum_{k=1}^N \sum_{i=1}^N p_{ik}(F(k)-r(i)),
\end{equation*}
and the persuasive constraints \eqref{eq:move-persuasive}, \eqref{eq:stay-persuasive} can be written as
\begin{align}
& \sum_{k=1}^N p_{ik}(F(k)-r(i)) \geq 0, \quad  i\in [N], \label{eq:move-reform}  \\
& \sum_{k=1}^{N-1}\left(\frac{1}{k}\sum_{j=1}^N p_{jk} -p_{ik}\right)(F(k+1) - r(i)) \notag \\
& + \left(1 - \sum_{k=1}^N \frac{1}{k}\sum_{j=1}^N p_{jk}\right)(F(1) - r(i))  \leq \frac{\mu(0)}{\mu(1)}r(i)  , \quad  i\in [N]. \label{eq:stay-reform}
\end{align}
\end{lemma}

Our main result in this section shows the converse is also true: 
given $p=\{p_{ik} > 0: i \in [N], k \in [N]\}$ satisfying persuasive constraints \eqref{eq:move-reform} and \eqref{eq:stay-reform}, with a few more linear constraints ensuring the $p_{ik}$'s are valid joint probabilities, there exists a persuasive signaling scheme $\phi$ satisfying Lemma~\ref{lem:no-move-zero-theta} such that \eqref{eq:pik-def} holds. Furthermore, there exists a polynomial time sequential sampling procedure that samples the set of agents to recommend to move as per the signaling scheme $\phi$.

\begin{lemma}
\label{lem:marginal-conditions}
Assume $p = \{p_{ik}>0: i,k \in [N]\}$ satisfies the persuasive constraints \eqref{eq:move-reform} and \eqref{eq:stay-reform}. If $p$ further satisfies
\begin{align}
& \sum_{k=1}^N \frac{1}{k}\sum_{i=1}^N p_{ik} \leq 1, \label{eq:card-dist} \displaybreak[3]\\
& kp_{ik} \leq   \sum_{j=1}^N p_{jk}, \quad  k\in [N], i\in [N],  \label{eq:valid-marginal}
\end{align}
then there exists a persuasive signaling scheme $\phi$ such that $\phi(\varnothing | 0) = 1$ and $\phi(\cdot | 1)$ satisfies \eqref{eq:pik-def}.
\end{lemma}

Note that if $p_{ik}$ is the joint probability that $k$ agents are recommended to move to $\ell_1$ and agent $i$ is among them, then $q_k \defeq  \frac{1}{k} \sum_{i=1}^N p_{ik}$
is the probability that $k$ agents are recommended to move. Upon letting $q_0 = 1 - \sum_{k=1}^N q_k$, we note that \eqref{eq:card-dist} ensures $\{q_k\}_{k=0}^N$ is a valid probability distribution over $\{0, \ldots, N\}$. On the other hand,
$q_{ik} \defeq \frac{p_{ik}}{q_k} = \frac{kp_{ik}}{\sum_{j=1}^N p_{jk}}$
is the conditional marginal probability that agent $i$ is recommended to move
given there are $k$ agents asked to move, and \eqref{eq:valid-marginal} ensures $q_{ik}$'s are valid probabilities. To summarize, \eqref{eq:card-dist} and \eqref{eq:valid-marginal} requires for each $k \in [N]$, the conditional marginal probabilities $\{q_{ik}\}_{i=1}^N$ are in the $k$-uniform matroid polytope.

We briefly describe the sketch of the proof of Lemma~\ref{lem:marginal-conditions}, and omit the details due to space limit. The main idea is that there exists a sampling procedure that samples a set of agents such that for each $i,k \in [N]$, $p_{ik}$ is the probability that $k$ agents are sampled and agent $i$ is among them.  Specifically, this sampling procedure first samples the size $k'$ of the output set according $\{q_k\}_{k=0}^N$. Following that, given any $k' > 0$, we adopt a sequential sampling subroutine presented by \cite{tille1996elimination} to sample $k'$ agents. This subroutine eliminates one agent from agents $1, \ldots, N$ at each step, and ensures the probability a particular agent $i$ still remains in the pool after each step is either 1, or strictly less than 1 while proportional to $p_{ik'}$. If \eqref{eq:valid-marginal} is satisfied, when $N-k'$ agents are eliminated, the probability each agent $i$ remains in the output set is ensured by this subroutine to be $q_{ik}$, and the joint probability that $k'$ agents are sampled and agent $i$ is included is then $p_{ik'}$. We briefly describe this subroutine in Appendix~\ref{ap:appendix-a} and omit its details and proof of correctness, and refer interested readers to \cite{tille1996elimination}. We then let $\phi(\cdot | 1)$ be the probability distribution of the set sampled according to this procedure, and note that $\phi(\cdot | 1)$ thereby satisfies \eqref{eq:pik-def} and the persuasive constraints. 

Summarizing the preceding discussion, from Lemma~\ref{lem:reform-lp} and Lemma~\ref{lem:marginal-conditions}, to obtain the optimal private signaling mechanism, we first solve the following linear program \ref{lp:reformulate} to obtain the optimal solution $p^*$.
\begin{align}
\max_{p_{ik}: i, k \in [N]} \quad & \sum_{k=1}^N \sum_{i=1}^N p_{ik}(F(k)-r(i))  \notag \\
s.t. \quad & \sum_{k=1}^N p_{ik}(F(k)-r(i)) \geq 0, \quad  i\in [N],   \notag \\
& \sum_{k=1}^{N-1}\left(\frac{1}{k}\sum_{j=1}^N p_{jk} -p_{ik}\right)(F(k+1) - r(i)) \notag \\
\lplabel[lp:reformulate]{(LP.2)} & + \left(1 - \sum_{k=1}^N \frac{1}{k}\sum_{j=1}^N p_{jk}\right)(F(1) - r(i))  \leq \frac{\mu(0)}{\mu(1)}r(i)  , \quad  i\in [N], \notag  \\
& \sum_{k=1}^N \frac{1}{k}\sum_{i=1}^N p_{ik} \leq 1, \notag \displaybreak[3] \\
& kp_{ik} \leq   \sum_{j=1}^N p_{jk}, \quad  k\in [N], i\in [N],  \notag \\
 & p_{ik} \geq  0, \quad k\in [N], i\in [N], \notag
\end{align}
After obtaining $p^*$, we sample the set of agents to recommend to move according to the following procedure given in Algorithm~\ref{alg:alg-private}. Let $\phi^*$ be the persuasive signaling mechanism corresponding to $p^*$ and Algorithm~\ref{alg:alg-private}, as given by Lemma~\ref{lem:marginal-conditions}. $\phi^*$ must be feasible for \ref{eq:lp-private}. Furthermore, the objective corresponding to $\phi^*$ in \ref{eq:lp-private} is equal to that of $p^*$ in \ref{lp:reformulate}. Therefore $\phi^*$ must be the optimal solution to \ref{eq:lp-private}, and is the optimal private signaling mechanism. 

Note that \ref{lp:reformulate} has $N^2$ variables and $N^2+2N+1$ constraints, therefore can be solved efficiently in polynomial time. The sequential sampling subroutine in \cite{tille1996elimination} takes at most $N$ steps to sample the set of agents to move, and in each step it takes polynomial time to determine which agent to eliminate. Therefore, the entire process to obtain the set of agents to recommend to move under the optimal private signaling mechanism can be completed in polynomial time.

\begin{algorithm}[h]
\KwIn{$p = \{p_{ik}: i \in [N], k \in [N]\}$ satisfying \eqref{eq:card-dist} and \eqref{eq:valid-marginal}. }
\KwOut{$\setS$.}

\For{$k\gets1$ \KwTo $N$}{$q_k \gets \frac{\sum_{i=1}^Np_{ik}}{k}$,\ } 

$q_0 \gets 1 - \sum_{k=1}^N q_k$\ 

   Sample the size of the output set: $k'$, according to $\{q_k\}_{k=0}^N$.\
   
   \eIf{$k'>0$}{Sample $k'$ agents from all agents, satisfying the marginals $\{q_{ik'}\}_{i=1}^N$ in the $k'$-uniform matroid polytope, where $q_{ik'}=\frac{p_{ik'}}{q_{k'}}$.\
   
   Let $\setS$ be the set of sampled agents.
   }{$\setS \gets \varnothing$.}
 
    \caption{Sampling the set of Agents to Move.}
    \label{alg:alg-private}
\end{algorithm}

We summarize our main result in this section as the following theorem.
\begin{theorem}
The optimal private signaling mechanism can be computed by first solving \ref{lp:reformulate}, then sampling the set of agents to recommend to move according to Algorithm~\ref{alg:alg-private}.
\end{theorem}

To conclude this section, we provide an observation that under certain conditions of modeling parameters, recommending the agents to follow the social optimal strategy profile is persuasive.  

Let $\tilde{W}(n) \defeq nF(n) - \sum_{i=1}^n r(i)$ be the social welfare when $\theta=1$ and the first $n$ agents move to $\ell_1$. Since $nF(n)$ is concave and $r(i)$ is increasing in $i$, $\tilde{W}(n)$ is also concave in $n$. Let $i^* \defeq \max\{i: i \in \argmax_{0 \leq i' \leq N} \tilde{W}(i') \}$ be the largest maximizer of $\tilde{W}$.  In the following proposition, we show the strategy profile that all agents staying when $\theta=0$, and the first $i^*$ agents moving when $\theta=1$, achieves the highest social welfare among all strategy profiles, and we give a sufficient condition under which recommending the agents to follow this strategy profile is persuasive.

\begin{proposition}
\label{thm:private-condition}
If $\mu(1)  \leq r(i^*+1)/F(i^*+1)$, then recommending all agents to stay when $\theta=0$ and recommending the first $i^*$ agents to move when $\theta=1$ is an optimal persuasive mechanism.
\end{proposition}

\textbf{Remark.}  Note that $r(i^*+1)/F(i^*+1)$ does not depend on
$\mu$. Thus, this proposition gives an upper bound on $\mu(1)$ for
achieving the maximum possible social welfare using signaling mechanisms. In
Appendix~\ref{ap:mu1-upper} we compute this upper bound for $\mu(1)$
under several typical families of resource sharing functions and cost
structures. Also, since $\tilde{W}$ is concave, computing for $i^*$
takes only $O(\log N)$ time, reducing the computational time for the
optimal mechanism substantially compared with the general method.

\section{Public Signaling Mechanism}
\label{sec:public-signal}

Having characterized the optimal private signals, we next consider the principal's problem under the restriction that the signaling mechanism be public. One reason for studying this restriction is that private signaling makes strong assumptions of no information leakage among the agents, an assumption that can easily fail in practice. Public signaling by construction avoids this information leakage concern. From a practical standpoint, restriction to public signaling can arise due to fairness requirements, where the principal seeks to avoid the semblance of providing conflicting information to different agents. 

A main technical difficulty in analyzing public signaling is the failure of the revelation principle argument\cite{kamenica2011bayesian, bergemann2019information}: it no longer suffices to optimize only over on straightforward and persuasive \emph{public} mechanisms. To overcome this difficulty, we first note that in a public signaling mechanism, after any information transmission from the principal, all the agents have a common belief about the resource level $\theta$, and participate in a Bayesian game under this common belief, where each agent simultaneously chooses whether to move to $\ell_1$. Thus, we begin our analysis of public signaling mechanisms by first analyzing the structure of the equilibria of this Bayesian game under any common belief of the agents. 


\subsection{Equilibrium structure}


Subsequent to receiving a public signal, let $q \in [0,1]$ denote the {\em common belief} of the agents that $\theta = 1$. The Bayes-Nash equilibrium of the subsequent game can be represented by a strategy profile $p = (p_1, \cdots, p_N)$, where $p_i$ denotes the probability that agent $i$ chooses to move to $\ell_1$. Our goal is to identify, for any common belief $q \in [0,1]$, the equilibrium profile that achieves the highest expected social welfare, given by $W(q,p) = q \expec[W(1,a) | a \sim p] = q \expec[ n(a) F(n(a)) | a \sim p] - \sum_i p_i r_i$. 

We first note that the set of equilibria of this game can be quite complex. First, there can be a multiplicity of equilbria for any $q \in [0,1]$, as the following example illustrates for an instance with $2$ agents:

\textbf{Example:} Let $N=2$, $F(1) = 1$, $F(2) = 0.6$, $r(1) = 0.5$ and $r(2) = 0.6$. For $q \in [0,0.5)$, the only equilibrium is $p(q) = (0,0)$. For $q \in [0.5, 0.6)$, the only equilibrium is $p(q) = (1,0)$. For $q \in [0.6, 5/6]$, there are three equilibria: $p_1(q) = (1,0)$, $p_2(q) = (0,1)$ and $p_3(q) = ((5q-3)/2q,(10q-5)/4q)$. For $q \in (5/6,1)$, there is a unique equilibrium $(1,0)$ and for $q=1$, there are two equilibria, $(1,0)$ and $(1,1)$. 

Secondly, as the following proposition shows, the equilibria themselves have very counterintuitive features, where if multiple agents randomize, then the one with larger moving costs must move with a higher probability:
\begin{proposition}
\label{thm:public-mixed-monotone}
For an equilibrium profile $(p_1, \ldots, p_N)$ under a common belief $q$, let $\imix \subseteq [N]$ be the set of agents who randomize: $\imix\defeq \{i \in [N]: 0 < p_i <1 \}$. For any agent $i, j\in \imix$ where $i < j$, it must be that that $p_i \leq  p_j$.
\end{proposition}
While the preceding result goes counter to the intuition that agents with higher moving cost should be less likely to move, it is explained by the fact that in order for an agent with higher moving cost to be indifferent between moving and staying, it must be that she must find location $\ell_1$ to be less competitive than the one with the lower moving cost, which implies that the one with the lower moving cost must move to $\ell_1$ with lower probability. 

Despite the complexity of the equilibrium set, one can nevertheless consider a simple class of equilibria, namely the {\em threshold} equilibria, defined as follows:
\begin{definition} For common belief $q \in [0,1]$, an equilibrium $p = (p_1, \cdots, p_N)$ is said to be a threshold equilibrium, if there exists a $t \in [0,N]$ such that
\begin{align*}
 p_i=\begin{cases} 1, & \text{if $i < \lceil t \rceil$;} \\
t + 1- \lceil t \rceil, & \text{if $i = \lceil t \rceil$;} \\
0, & \text{if $i> \lceil t \rceil$},
\end{cases}
\end{align*}
where $ \lceil t \rceil$ is the smallest integer that is greater than or equal to $t$. We denote such an equilibrium by $(q,t)$.    
\end{definition}
In a threshold equilibrium $(q,t)$, at most one agent randomizes between moving and staying. Furthermore, all agents $i < \lceil t \rceil$ move to $\ell_1$, whereas all agents $i > \lceil t \rceil$ stay at $\ell_0$. Thus, threshold equilibria capture the intuition that agents with higher moving costs should move with lower probability. Our first result shows that, indeed, for any $q \in [0,1]$, there exists a threshold equilibrium. Before we state our result, we define the following quantities for any common belief $q$:
\begin{equation}
\label{eq:ihat-def}
\begin{split}
\uihat(q) \defeq &\max\{i \in \{0, \ldots, N\}: qF(i)  - r(i) \geq 0\}, \\
\lihat(q) \defeq & \max\{i \in \{0, \ldots, N\}: qF(i) - r(i) > 0\}.
\end{split}
\end{equation} 
Note that $q F(i) - r(i)$ denotes the expected utility of agent $i$ for moving to location $\ell_1$, if all agents $j<i$ move to $\ell_1$ and all agents $j>i$ stay in $\ell_0$. Thus, $\lihat(q)$ denotes the agent with the largest moving cost who strictly prefers to move to $\ell_1$, if all agents with smaller moving costs move, and those with larger moving costs stay. Similarly, $\uihat(q)$ denotes the agent with the largest moving cost who does not strictly prefer to stay, under similar choices of other agents. Moreover, since $F(i)$ is decreasing in $i$ and $r(i)$ is increasing in $i$, we have $\uihat(q) \geq \lihat(q)$ for any belief $q$. Furthermore, since $qF(0) - r(0) \geq 0$, we have $\uihat(q) \geq 0$ and hence $[\lihat(q) , \uihat(q)] \cap [0,N] \neq \emptyset$. We have the following result:
\begin{lemma}
\label{lem:threshold-equilibrium}
For any $q \in [0,1]$, and any $t \in [0,N]$, $(q,t)$ is a threshold equilibrium if and only if $t \in [\lihat(q), \uihat(q)]$.
\end{lemma}
The preceding lemma guarantees the existence of threshold equilibria. The question then is how do threshold equilibria fare against other equilibria in terms of their expected social welfare. The following result, our main theorem of this section, establishes that for any common belief, the expected social welfare over all equilibria is attained at a threshold equilibrium.
\begin{theorem}
\label{thm:threshold-optimal}
For any common belief $q \in [0,1]$, the expected social welfare $W(q,p)$ under any equilibrium $(q,p)$ is no more than that under the threshold equilibrium $(q, \lihat(q))$: 
\[W(q, \lihat(q)) \geq W(q,p), \quad \text{for any equilibrium $p$}.\]
\end{theorem}
The preceding theorem has the following implications: First, the expected social welfare is always maximized at a pure equilibrium $(q, \lihat(q))$: no agent strictly randomizes between moving and staying. Second, under the optimal {\em public} signaling mechanism $(\Sigma, \phi)$, for any signal $s$ and the induced common belief $q \in [0,1]$, it must be that the resulting equilibrium among the agents is $(q, \lihat(q))$. For otherwise, one can always (publicly) recommend the agents to move to this equilibrium and improve the expected social welfare. In the next section, we use these two facts to completely characterize the optimal public signaling mechanism as a solution to linear program.

\subsection{Optimal Public Signaling Mechanism}

The results in the preceding section imply the following structure for the optimal public signaling mechanism: For each $\theta$, the principal (publicly) recommends the threshold $i$ number of agents to move to location $\ell_1$, with the constraint that, under the induced common belief $q$ upon receiving the recommendation $i$, it must be the case that $\lihat(q) = i$. This follows from an argument similar to the revelation-principle style argument for the private signaling mechanism, with the modification where the condition $\lihat(q) = i$ plays the same role of the persuasive constraints. We omit the details due to its similarity to that of the private signaling case.

Thus, the public signaling mechanism can be described by choosing the signal set to be $\Sigma = \{0,\cdots, N\}$ and $\phi : \Theta \to \Delta(\Sigma)$, where $\phi(i|\theta)$ denotes the probability that the principal recommends the first $i$ agents to move to $\ell_1$ when the resource level is $\theta$. The common belief upon sending a public signal $s=i$ is then given by $q_i = \mu(1)\phi(i|1)/ (\mu(1) \phi(i|1) + \mu(0)\phi(i|0))$. The equilibrium constraint is thus $\lihat(q_i) = i$ for each $i \in \Sigma$. Note that the definition of $\lihat(q)$ then implies that 
\begin{align}\label{eq:lihat-cond}
    q_i F(i) - r(i) &> 0\\
    q_i F(i+1) - r(i+1) &\leq 0.
\end{align}
Finally, the expected social welfare is given by $\sum_{\theta = 0,1}  \mu(\theta)  \sum_{i=0}^N \phi(i|\theta) W(\theta, i)$, where $W(\theta,i) = \theta iF(i) - \sum_{j \leq i} r_j$. 
Taken together, using the fact that $q_i = \mu(1)\phi(i|1)/ (\mu(1) \phi(i|1) + \mu(0)\phi(i|0))$ and letting $\phi(\theta, i) = \mu(\theta)\phi(i|\theta)$, we obtain the following LP in $\phi$:
\begin{equation}
\label{eq:public-lp}
\begin{split}
\max_{\phi} \quad  \sum_{\theta =0, 1}\sum_{i =0}^N  W(\theta, i) \cdot \phi( \theta,i)  \\
\text{s.t.} \quad  \left(F(i) - r(i)\right)\cdot \phi(1,i) - r(i)\cdot \phi(0,i) &\geq 0, \qquad \text{for all $i \in \{1, \cdots, N\}$;} \\
 \left(F(i+1) - r(i+1)\right)\cdot \phi(1,i) -r(i+1)\cdot\phi(0,i) &\leq 0, \quad \text{for all $i \in \{0, \cdots, N-1\}$;} \\
 \sum_{i =0, \ldots, N} \phi(\theta, i) = \mu(\theta), \quad 
 \phi(\theta, i) &\geq 0, \quad \text{for $\theta =0, 1$,  for all $i \in \{0, \cdots, N\}$.}
\end{split}
\end{equation}
In \eqref{eq:public-lp}, the objective is the expected social welfare under $\phi$. The first two sets of constraints encode \eqref{eq:lihat-cond}, ensuring that each signal $i$ indeed induces the corresponding threshold equilibrium. (Note that since $W(\theta, i)$ is decreasing in $i$, the strict inequality in \eqref{eq:lihat-cond} can be replaced by the weaker inequality without loss of optimality.) The last constraints ensure $\phi(\cdot |\theta)$ is a valid probability distribution for each $\theta$. Note that the preceding LP has $2(N+1)$ variables and $2(N+1)$ constraints, and thus can be efficiently solved in polynomial time.

\textbf{Remark}. With a similar argument as that in \cite{kamenica2011bayesian}, one can show the optimal public signaling mechanism sends at most two signals with positive probabilities. That is, the optimal public signaling mechanism randomizes over two thresholds, with different weights for $\theta=0$ and $\theta=1$. Our results in this section can be easily extended to allow for larger state space $\Theta$, with cardinality any finite $K$. In this case, the optimal public signal is the solution to a linear program with $O(NK)$ number of variables and constraints, and the optimal public scheme sends at most $K$ signals with positive probability.

\section{Computational Results}

\begin{figure}[h!]
     \centering
     \begin{subfigure}[b]{\textwidth}
         \centering
         \includegraphics[width=\textwidth]{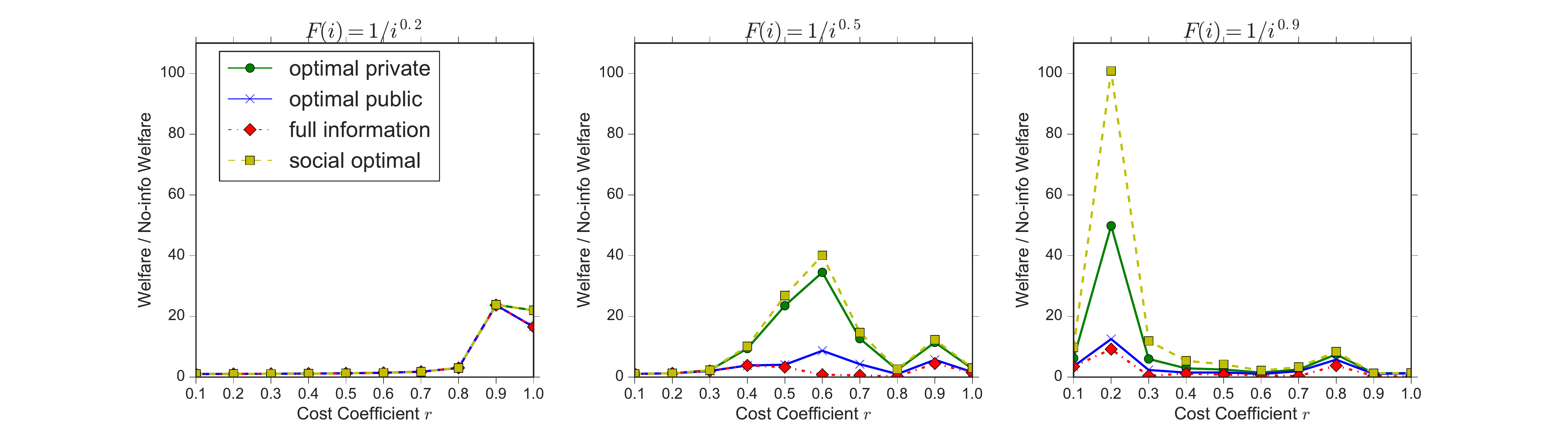}
         \caption{Constant costs: $r(i)=0.5r$.}
         \label{fig:ratio-constant}
     \end{subfigure}
     \begin{subfigure}[b]{\textwidth}
         \centering
         \includegraphics[width=\textwidth]{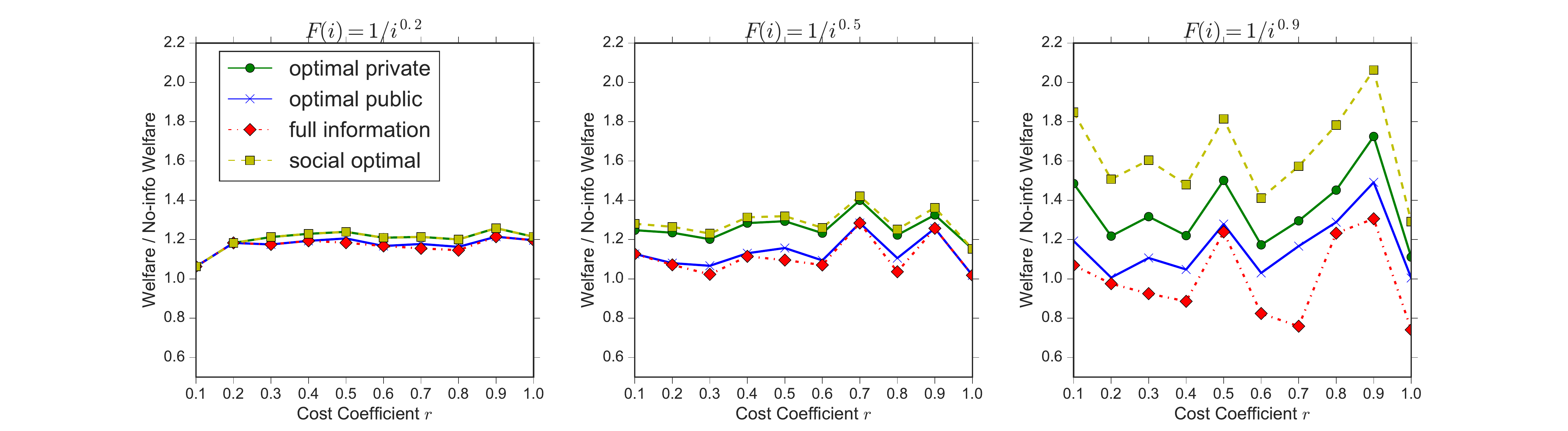}
         \caption{Linear costs: $r(i)=0.1ri$.}
         \label{fig:ratio-linear}
     \end{subfigure}
     \begin{subfigure}[b]{\textwidth}
         \centering
         \includegraphics[width=\textwidth]{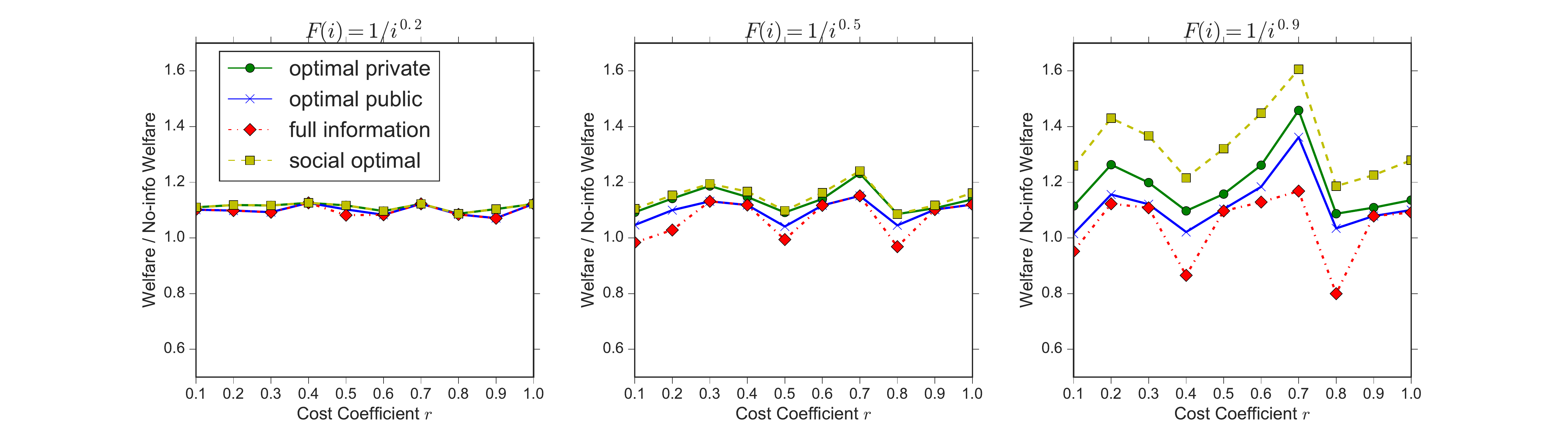}
         \caption{Quadratic costs: $r(i)=0.02ri^2$.}
         \label{fig:ratio-quad}
     \end{subfigure}
        \caption{Social welfare of the optimal signaling mechanisms and the benchmarks. $x$-axis is the cost coefficient $r$, and $y$-axis is the social welfare of different mechanisms/benchmarks over social welfare under no information sharing. For all cases, the total number of agents $N=20$ and the prior is $\mu(0)=0.2$, $\mu(1)=0.8$.}
  \label{fig:social-welfare-numeric}
\end{figure}

In this section, we compare numerically the social welfare under the optimal private and public signaling mechanisms with three benchmarks: the no-information benchmark, the full-information benchmark, and the social optimal benchmark, where all agents choose the social optimal action, that is, all agents stay when $\theta=0$, and the first $i^*$ agents move when $\theta=1$ and $i^*$ is defined earlier as the largest maximizer of $\tilde{W}(n) = nF(n) - \sum_{i=1}^n r(i)$. Note the third benchmark may not be achievable through a signaling mechanism, but rather provides an upper bound for the social welfare any signaling mechanism can generate.

We consider three different resource sharing function $F(i)=1/i^\alpha$ for $\alpha=0.2, 0.5, 0.9$. The parameter $\alpha$ controls the curvature of the resource sharing function, representing different resource sharing scenarios. We consider three different cost structures, constant costs where $r(i)=0.5r$, linear costs where $r(i)=0.1ri$ and quadratic costs where $r(i)=0.02ri^2$, where $r$ is the cost coefficient and we vary $r$ to represent different cost functions. We compute the social welfare of the optimal mechanisms and the benchmarks for different resource sharing and cost functions and plot the results in Figure~\ref{fig:social-welfare-numeric}.

Numerical results in Figure~\ref{fig:social-welfare-numeric} show that the optimal private signaling mechanism generates a social welfare close to the social optimal benchmarks in most cases, and the optimal public signaling mechanism also generates social welfare higher than the full information or no information benchmark in most scenarios. The benefit of private signaling is high when agents have similar costs, as illustrated by the subplots in the first column. \cite{arieli2019private} and \cite{candogan2017optimal} show that when agents' actions have positive externalities, the optimal private mechanism should correlate the recommendations to take the positive action for every agent as much as possible. Here in our setting where agents' actions have negative externalities, a different phenomenon is observed: private signaling  de-correlates the agents' recommendations and generates more expected social welfare compared to public signaling where agents' actions are more correlated.

In Proposition~\ref{thm:private-condition}, we give an upper bound for $\mu(1)$ under which the social optimal benchmark is achievable via private signaling. We also study numerically how much social welfare can be generated when this condition does not hold, by considering different priors. We presents these results in Appendix~\ref{ap:numeric-mu}.

\section{Discussion}
\label{sec:conclusion-signal}

We consider information design in a two location resource competition model. For private signaling, we provide a method to compute the optimal mechanism in polynomial time, and also characterize a condition of model parameters under which the optimal mechanism has a simple threshold structure. For public signaling mechanisms, we characterize the structure of the socially optimal equilibrium and establish the form of the optimal public signaling mechanism. Numerical results show the optimal private and public signaling mechanisms increase the social welfare substantially compared with the no-information and full-information setting.

Readers may have noticed that our method for obtaining the optimal private signaling mechanism does not restrict to the resource sharing setting we are considering. In fact, it is applicable to all settings where the externality of each agent's action is anonymous (the utility of an agent depends only on how many other agents are taking the same action, but not on which agents are taking this action). Formally, for information design settings with binary action spaces where receiver $i$'s utility has form $f_i\left(a_i, \sum_{j \ne i} a_j\right)$ and the sender's utility has form $f\left(\sum_i a_i\right) + \sum_i g_i(a_i)$, our method can be used to obtain the optimal private signaling mechanism in polynomial time. Such utility functions can be found in many settings, for example, in most voting settings where the sender does not differentiate across voter.

Going beyond our model, in many spatial settings, there are more than two locations. The form of signals and the action and belief structure of the agents can be much more complicated. Fully solving the optimal signaling mechanism in such a system may be computationally intractable, while it is hopeful heuristic mechanisms with theoretical guarantee or good empirical performance on the social welfare generated may exist and is an interesting future research path. Meanwhile, in most settings, agents migrate among locations so signaling is not ``one-shot'' but rather a sequential and dynamic process, where signals sent at a time impact the actions and beliefs of agents ever since hence affect signals that should be sent later. The Bayesian persuasion framework does not extend naturally to these settings. Building models and analyze such dynamic signaling process is an important future path for better understanding real world spatial signaling problems.


\begin{thebibliography}{10}

\bibitem{alonso2016persuading}
Ricardo Alonso and Odilon C{\^a}mara.
\newblock Persuading voters.
\newblock {\em American Economic Review}, 106(11):3590--3605, 2016.

\bibitem{arieli2019private}
Itai Arieli and Yakov Babichenko.
\newblock Private bayesian persuasion.
\newblock {\em Journal of Economic Theory}, 182:185--217, 2019.

\bibitem{babichenko2016computational}
Yakov Babichenko and Siddharth Barman.
\newblock Computational aspects of private bayesian persuasion.
\newblock {\em arXiv preprint arXiv:1603.01444}, 2016.

\bibitem{babichenko2017algorithmic}
Yakov Babichenko and Siddharth Barman.
\newblock Algorithmic aspects of private bayesian persuasion.
\newblock In {\em 8th Innovations in Theoretical Computer Science Conference
  (ITCS 2017)}. Schloss Dagstuhl-Leibniz-Zentrum fuer Informatik, 2017.

\bibitem{badanidiyuru2018targeting}
Ashwinkumar Badanidiyuru, Kshipra Bhawalkar, and Haifeng Xu.
\newblock Targeting and signaling in ad auctions.
\newblock In {\em Proceedings of the Twenty-Ninth Annual ACM-SIAM Symposium on
  Discrete Algorithms}, pages 2545--2563. SIAM, 2018.

\bibitem{bardhi2018modes}
Arjada Bardhi and Yingni Guo.
\newblock Modes of persuasion toward unanimous consent.
\newblock {\em Theoretical Economics}, 13(3):1111--1149, 2018.

\bibitem{bergemann2015limits}
Dirk Bergemann, Benjamin Brooks, and Stephen Morris.
\newblock The limits of price discrimination.
\newblock {\em American Economic Review}, 105(3):921--57, 2015.

\bibitem{bergemann2016bayes}
Dirk Bergemann and Stephen Morris.
\newblock Bayes correlated equilibrium and the comparison of information
  structures in games.
\newblock {\em Theoretical Economics}, 11(2):487--522, 2016.

\bibitem{bergemann2019information}
Dirk Bergemann and Stephen Morris.
\newblock Information design: A unified perspective.
\newblock {\em Journal of Economic Literature}, 57(1):44--95, 2019.

\bibitem{caillaud2007consensus}
Bernard Caillaud and Jean Tirole.
\newblock Consensus building: How to persuade a group.
\newblock {\em American Economic Review}, 97(5):1877--1900, 2007.

\bibitem{candogan2019persuasion}
Ozan Candogan.
\newblock Persuasion in networks: Public signals and k-cores.
\newblock {\em Available at SSRN}, 2019.

\bibitem{candogan2017optimal}
Ozan Candogan and Kimon Drakopoulos.
\newblock Optimal signaling of content accuracy: Engagement vs. misinformation.
\newblock {\em Misinformation (October 11, 2017)}, 2017.

\bibitem{chakraborty2014persuasive}
Archishman Chakraborty and Rick Harbaugh.
\newblock Persuasive puffery.
\newblock {\em Marketing Science}, 33(3):382--400, 2014.

\bibitem{crawford1982strategic}
Vincent~P Crawford and Joel Sobel.
\newblock Strategic information transmission.
\newblock {\em Econometrica: Journal of the Econometric Society}, pages
  1431--1451, 1982.

\bibitem{drakopoulos2018persuading}
Kimon Drakopoulos, Shobhit Jain, and Ramandeep~S Randhawa.
\newblock Persuading customers to buy early: The value of personalized
  information provisioning.
\newblock 2018.

\bibitem{dughmi2017survey}
Shaddin Dughmi.
\newblock Algorithmic information structure design: a survey.
\newblock 15(2):2--24, 2017.

\bibitem{dughmi2017algorithmic}
Shaddin Dughmi and Haifeng Xu.
\newblock Algorithmic persuasion with no externalities.
\newblock In {\em Proceedings of the 2017 ACM Conference on Economics and
  Computation}, pages 351--368. ACM, 2017.

\bibitem{farrell1987cheap}
Joseph Farrell.
\newblock Cheap talk, coordination, and entry.
\newblock {\em The RAND Journal of Economics}, pages 34--39, 1987.

\bibitem{gentzkow2017bayesian}
Matthew Gentzkow and Emir Kamenica.
\newblock Bayesian persuasion with multiple senders and rich signal spaces.
\newblock {\em Games and Economic Behavior}, 104:411--429, 2017.

\bibitem{grossman1981informational}
Sanford~J Grossman.
\newblock The informational role of warranties and private disclosure about
  product quality.
\newblock {\em The Journal of Law and Economics}, 24(3):461--483, 1981.

\bibitem{hssaine2018information}
Chamsi Hssaine and Siddhartha Banerjee.
\newblock Information signal design for incentivizing team formation.
\newblock {\em arXiv preprint arXiv:1809.00751}, 2018.

\bibitem{kamenica2011bayesian}
Emir Kamenica and Matthew Gentzkow.
\newblock Bayesian persuasion.
\newblock {\em American Economic Review}, 101(6):2590--2615, 2011.

\bibitem{kremer2014implementing}
Ilan Kremer, Yishay Mansour, and Motty Perry.
\newblock Implementing the “wisdom of the crowd”.
\newblock {\em Journal of Political Economy}, 122(5):988--1012, 2014.

\bibitem{lingenbrink2017queue}
David Lingenbrink and Krishnamurthy Iyer.
\newblock Optimal signaling mechanisms in unobservable queues with strategic
  customers.
\newblock In {\em Proceedings of the 2017 ACM Conference on Economics and
  Computation}, pages 347--347. ACM, 2017.

\bibitem{lingenbrink2018retail}
David Lingenbrink and Krishnamurthy Iyer.
\newblock Signaling in online retail: Efficacy of public signals.
\newblock In {\em Proceedings of the 13th Workshop on Economics of Networks,
  Systems and Computation}, NetEcon '18, pages 10:1--10:1, New York, NY, USA,
  2018. ACM.

\bibitem{mathevet2017information}
Laurent Mathevet, Jacopo Perego, and Ina Taneva.
\newblock On information design in games.
\newblock {\em Unpublished paper, Department of Economics, New York
  University.[1115]}, 2017.

\bibitem{milgrom1981good}
Paul~R Milgrom.
\newblock Good news and bad news: Representation theorems and applications.
\newblock {\em The Bell Journal of Economics}, pages 380--391, 1981.

\bibitem{rayo2010optimal}
Luis Rayo and Ilya Segal.
\newblock Optimal information disclosure.
\newblock {\em Journal of political Economy}, 118(5):949--987, 2010.

\bibitem{taneva2015information}
Ina~A Taneva.
\newblock Information design.
\newblock 2015.

\bibitem{tille1996elimination}
Yves Till{\'e}.
\newblock An elimination procedure for unequal probability sampling without
  replacement.
\newblock {\em Biometrika}, 83(1):238--241, 1996.

\bibitem{ubersurge}
Uber.
\newblock How surge pricing works, 2019.

\bibitem{wang2013bayesian}
Yun Wang.
\newblock Bayesian persuasion with multiple receivers.
\newblock {\em Available at SSRN 2625399}, 2013.

\bibitem{xu2016signaling}
Haifeng Xu, Rupert Freeman, Vincent Conitzer, Shaddin Dughmi, and Milind Tambe.
\newblock Signaling in bayesian stackelberg games.
\newblock In {\em Proceedings of the 2016 International Conference on
  Autonomous Agents \& Multiagent Systems}, pages 150--158. International
  Foundation for Autonomous Agents and Multiagent Systems, 2016.

\bibitem{xu2015exploring}
Haifeng Xu, Zinovi Rabinovich, Shaddin Dughmi, and Milind Tambe.
\newblock Exploring information asymmetry in two-stage security games.
\newblock In {\em Twenty-Ninth AAAI Conference on Artificial Intelligence},
  2015.

\end{thebibliography}

\appendix

\section{Proofs}
\label{ap:proofs}

\begin{proof}[\textbf{Proof of Lemma~\ref{lem:no-move-zero-theta}}] 
$ $\newline

\noindent
Let $\phi$ be an arbitrary persuasive straightforward signaling scheme, where for any $S \subseteq [N]$, $\phi(S | \theta)$ is the probability of recommending agents in $S$ to move, condition on resource state is $\theta$. We construct another straightforward signaling scheme $\phi'$ where $\phi'(S | 1) = \phi(S|1)$ for all $S  \subseteq [N]$; $\phi'(\varnothing | 0) = 1$ and  $\phi'(S | 0) = 0$ for all $S \subseteq [N]$ and $ S \ne \{\varnothing\}$. We would show $\phi'$ is also persuasive and achieves at least the same social welfare as $\phi$. The expected social welfare under $\phi$ is
\begin{align*}
& \sum_{\theta=0,1, S \subseteq [N]} \mu(\theta) \phi(S | \theta) W(\theta, S)\\
= & \sum_{\theta=0,1, S \subseteq [N]} \phi(\theta, S) W(\theta, S) \\
=&  \sum_{S \subseteq [N]} \phi(1, S) W(1, S) + \sum_{S \subseteq [N]} \phi(0, S)W(0, S) \\
=& \sum_{S \subseteq [N]} \phi(1, S) W(1, S) - \sum_{S \subseteq [N]} \phi(0, S) \sum_{i \in S} r(i) \displaybreak[3]\\
\leq & \sum_{S \subseteq [N]} \phi(1, S) W(1, S) \\
= & \sum_{S \subseteq [N]} \phi'(1, S) W(1, S) + \phi'(0, \varnothing) W(0, \varnothing) \\
= & \sum_{\theta=0,1, S \subseteq [N]} \phi'(\theta, S) W(\theta, S), 
\end{align*}
which is the social welfare under $\phi'$. 

Next, we show $(\Sigma, \phi')$ satisfies the persuasive constraints \eqref{eq:move-persuasive} and \eqref{eq:stay-persuasive}. For each agent $i$, since $\phi'(S | 0)=0$ for $S \ne \varnothing$ and $\phi'(S | 1)= \phi(S | 1)$ for any subset $S$, we have
\begin{align*}
& \sum_{\theta, S: i \in S} \phi'(\theta, S)(\theta F(|S|) - r(i))  \\
= & \sum_{S: i \in S} \phi'(1, S)(F(|S|) - r(i))  \\
= & \sum_{S: i \in S} \phi(1, S)(F(|S|) - r(i)) \displaybreak[3]\\
\geq &  -r(i) \sum_{S: i \in S} \phi(0, S) + \sum_{S: i \in S} \phi(1, S)(F(|S|) - r(i)) \\
= & \sum_{\theta, S: i \in S} \phi(\theta, S)(\theta F(|S|) - r(i)) \\
\geq & 0,
\end{align*}
where the last inequality is because $\phi$ is persuasive. Therefore we have shown \eqref{eq:move-persuasive} holds for agent $i$.

On the other hand, for each agent $i$, we have
\begin{align*}
 & \sum_{\theta, S: i\notin S} \phi'(\theta, S)(\theta F(|S|+1) - r(i)) \\
= & -r(i) \sum_{S: i \notin S} \phi'(0, S) + \sum_{S: i\notin S} \phi'(1, S)( F(|S|+1) - r(i)) \\
= & -r(i)\mu(0) + \sum_{S: i\notin S} \phi(1, S)( F(|S|+1) - r(i)) \\
\leq & -r(i) \sum_{S: i \notin S}  \phi(0, S) + \sum_{S: i\notin S} \phi(1, S)( F(|S|+1) - r(i)) \\
= & \sum_{\theta, S: i\notin S} \phi(\theta, S)(\theta F(|S|+1) - r(i)) \\ 
\leq &  0,
\end{align*}
where the first inequality is because $\sum_{S: i \notin S} \phi(0, S) \leq \mu(0)$, and the last inequality holds because $\phi$ is persuasive. Thus we have shown \eqref{eq:stay-persuasive} holds for agent $i$.

\end{proof}

\begin{proof}[\textbf{Proof of Lemma~\ref{lem:reform-lp}}] 
$ $\newline

\noindent
Let $\setS$ be the random subset of agents recommended to move under $\phi$. For each agent $i$, let $X_i = \mathds{1}\{i \in \setS\}$ be the random variable denoting that agent $i$ is recommended to move.  Let $\prob$ and $\expec$ denote the probability measure and expectation induced by $\phi$.

Since $p_{ik} = \prob(|\setS|=k, i \in \setS | \theta=1) = \prob(|\setS|=k | \theta=1) \prob(i \in \setS | |\setS|=k, \theta=1)$, 
and
$$k = \expec\left[\sum_{i=1}^N X_i \bigg| |\setS|=k, \theta=1 \right] = \sum_{i=1}^N \expec[X_i | |\setS|=k, \theta=1]  = \sum_{i=1}^N\prob(i \in \setS | |\setS|=k, \theta=1),$$
we have 
$$ \sum_{i=1}^N p_{ik} = \prob(|\setS| = k | \theta=1 ) \sum_{i=1}^N \prob(i \in \setS | |\setS|=k, \theta=1) = k \prob(|\setS|=k | \theta=1),$$
and therefore $\prob(|\setS|=k | \theta=1) = \sum_{i=1}^N p_{ik} / k$.
 
The objective \eqref{eq:old-objective} of \ref{eq:lp-private} can be written as
\begingroup
\allowdisplaybreaks[1]
\begin{align*}
&  \mu(0) \sum_{ S \subseteq [N]} \phi(S |\theta) W(0, S) +  \mu(1) \sum_{ S \subseteq [N]} \phi(S |\theta) W(1, S) \\
= & \mu(0)\phi(\varnothing | 0)W(0, \varnothing) + \mu(1)\sum_{ S \subseteq [N]} \phi(S|1) \left(|S| F(|S|) -\sum_{i \in S}r(i) \right) \\
= & \mu(1) \sum_{k=1}^N \sum_{S :|S|=k} \phi(S|1) \left(k F(k) -\sum_{i \in S}r(i) \right)\\
= & \mu(1)\sum_{k=1}^N kF(k)\sum_{S: |S|=k} \phi(S|1) - \mu(1) \sum_{k=1}^N \sum_{S : |S|=k} \phi(S|1) \sum_{i \in S} r(i) \\ 
= & \mu(1)\sum_{k=1}^N kF(k)\prob(|\setS|=k | \theta=1) - \mu(1) \sum_{k=1}^N  \sum_{i=1}^N r(i) \sum_{S: |S|=k, i\in S} \phi(S|1) \\
= &\mu(1) \sum_{k=1}^N kF(k) \sum_{i=1}^N \frac{p_{ik}}{k} - \mu(1) \sum_{k=1}^N \sum_{i=1}^N r(i) p_{ik} \\
= & \mu(1) \sum_{k=1}^N \sum_{i=1}^N p_{ik}(F(k) - r(i)).
\end{align*}
 
Next, we write the persuasive constraints in terms of the $p_{ik}$'s. For each agent $i$, the left hand side of \eqref{eq:move-persuasive} can be written as 
\begin{align*}
& \sum_{\theta=0,1} \mu(\theta) \sum_{ S: i\in S} \phi(S | \theta)(\theta F(|S|) - r(i))\\
= & \mu(1) \sum_{S : i \in S} \phi(S | 1)(F(|S|)- r(i)) \\
= & \mu(1) \sum_{k=1}^N \sum_{S : |S|=k, i \in S} \phi(S|1)(F(k)-r(i)) \\
= & \mu(1) \sum_{k=1}^N(F(k)-r(i)) \sum_{S : |S|=k, i \in S} \phi(S|1)\\
= & \mu(1) \sum_{k=1}^N p_{ik}(F(k)-r(i)).
\end{align*}
Therefore for agent $i$, constraint \eqref{eq:move-persuasive} is equivalent as
\begin{align*}
\sum_{k=1}^N p_{ik}(F(k) - r(i)) \geq 0.
\end{align*}
On the other hand, for each agent $i$, the left hand side of \eqref{eq:stay-persuasive} can be written as 
\begin{align*}
& \sum_{\theta=0,1} \mu(\theta) \sum_{S: i\notin S} \phi(S | \theta)(\theta F(|S|+1) - r(i)) \\
= & -\mu(0)r(i) + \mu(1)\sum_{S: i\notin S} \phi(S | 1)( F(|S|+1) - r(i)) \\
= &  -\mu(0)r(i) + \mu(1) \sum_{k=0}^{N-1} \sum_{S: |S|=k,  i\notin S} \phi(S|1) (F(k+1) - r(i)) \\
= & -\mu(0) r(i) + \mu(1)\sum_{k=0}^{N-1} \left( \sum_{S: |S|=k} \phi(S|1) (F(k+1) - r(i)) - \sum_{S: |S|=k, i \in S} \phi(S|1) (F(k+1) - r(i)) \right) \\
= & -\mu(0) r(i) + \mu(1)\sum_{k=0}^{N-1} (F(k+1) - r(i)) \sum_{S: |S|=k} \phi(S|1)  - \mu(1)\sum_{k=1}^{N-1}  (F(k+1) - r(i)) \sum_{S: |S|=k, i \in S} \phi(S|1) \\
= & -\mu(0) r(i) + \mu(1)\sum_{k=0}^{N-1} (F(k+1) - r(i))\prob(|\setS|=k | \theta=1) - \mu(1)\sum_{k=1}^{N-1}  (F(k+1) - r(i)) p_{ik} \\
= & -\mu(0) r(i) - \mu(1)\sum_{k=1}^{N-1}  (F(k+1) - r(i)) p_{ik}  + \mu(1)\sum_{k=1}^{N-1} (F(k+1) - r(i))\prob(|\setS|=k | \theta=1) \\
& + \mu(1)(F(1) - r(i))\prob(|\setS|=0 | \theta=1) \\
=& -\mu(0) r(i) + \mu(1)\sum_{k=1}^{N-1}  (F(k+1) - r(i))(\prob(|\setS|=k | \theta=1) - p_{ik}) \\
& + \mu(1)(F(1) - r(i))\left(1  -\sum_{k=1}^N \prob(|\setS|=k|\theta=1) \right) \\
= & -\mu(0) r(i) + \mu(1)\sum_{k=1}^{N-1}  (F(k+1) - r(i))\left(\sum_{i=1}^N \frac{ p_{ik}}{k}  - p_{ik}\right) \\
& + \mu(1)(F(1) - r(i))\left(1  -\sum_{k=1}^N \sum_{i=1}^N \frac{ p_{ik}}{k}  \right) \\
= & -\mu(0) r(i) +\mu(1) \left(\left(1 - \sum_{k=1}^N \frac{1}{k}\sum_{j=1}^N p_{jk}\right)(F(1) - r(i))  +\sum_{k=1}^{N-1}\left(\frac{1}{k}\sum_{j=1}^N p_{jk} -p_{ik}\right)(F(k+1) - r(i)) \right.
\end{align*}
Therefore for agent $i$, constraint \eqref{eq:stay-persuasive} is equivalent as
\begin{align*}
-\mu(0)r(i) + \mu(1) \left(\left(1 - \sum_{k=1}^N \frac{1}{k}\sum_{j=1}^N p_{jk}\right)(F(1) - r(i))  +\sum_{k=1}^{N-1}\left(\frac{1}{k}\sum_{j=1}^N p_{jk} -p_{ik}\right)(F(k+1) - r(i)) \right) \leq 0
\end{align*}

\end{proof}

\begin{proof}[\textbf{Proof of Lemma~\ref{lem:marginal-conditions}}] 
$ $\newline

\noindent
For any $k > 0$, \cite{tille1996elimination} presents a sequential elimination subroutine that eliminates one agent from agents $1, \ldots, N$ at each step, and outputS the remaining $k$ agents, ensuring the probability each agent $i$ is included in the sample is $kp_{ik}/\sum_{j=1}^N p_{jk}$, so long as \eqref{eq:valid-marginal} holds for $p_{ik}$'s.

Specifically, at each step $n=N-1, N-2, \ldots, k$, this subroutine computes $\pi(i | n)$ for each $i \in [N]$ in the following way. It first computes for each $i$ the quantity
$$ \frac{np_{ik}}{\sum_{j=1}^N p_{jk}} .$$
For each agent $i$ such that this quantity is larger than 1, it assigns this quantity to $\pi(i | n)$. It then recomputes this quantity for remaining agents, with the multiplier in the numerator being the number of remaining agents instead of $n$. This process is repeated until for all $i \in [N]$, $\pi(i| n)$ is in $[0, 1]$. After this process, some $\pi(i | n)$ is equal to 1 and others are strictly proportional to $p_{ik}$.

This subroutine then computes for each $i \in [N]$
\begin{align*}
    r_{ni} = \begin{cases}
     1 - \pi(i|n), & \text{if $n=N-1$;}\\
    1 - \frac{\pi(i|n)}{\pi(i | n+1)}, & \text{if $n < N-1$.}
    \end{cases}
\end{align*}

For each $n$, let $S_n$ be the set of remaining agents at the beginning of step $n$, it can be verified $\{r_{ni}\}_{i \in S_n}$ is a valid probability distribution, and this subroutine samples one agent according to this distribution and eliminate this agent from the pool. If \eqref{eq:valid-marginal} holds, this subroutine ensures after the final step $k$, the probability each agent $i$ is included in the sample is $kp_{ik}/\sum_{j=1}^Np_{jk} = q_{ik}$.

The rest of the proof of this lemma is given in the main paper.

\end{proof}

\endgroup

\begin{proof}[\textbf{Proof of Proposition~\ref{thm:private-condition}}] 
$ $\newline

\noindent
First we show $F(i^*) \geq r(i^*)$. This claim holds trivially when $i^*=0$. Assume $i^* \geq 1$ and assume for contradiction $F(i^*) < r(i^*)$. We have
\begin{align*}
& \tilde{W}(i^*) - \tilde{W}(i^*-1) \\
= &  i^*F(i^*) - \sum_{j=1}^{i^*} r(j) - \left(  (i^*-1)^*F(i^*-1) - \sum_{j=1}^{i^*-1} r(j) \right) \\
= & (i^*-1)(F(i^*) - F(i^*-1)) + F(i^*) - r(i^*) \\
\leq & F(i^*) - r(i^*) \\
< & 0.
\end{align*}
Therefore $i^*$ is not maximizer of $\tilde{W}$, leading to a contradiction.

Let $x$ be the probability distribution over the subsets of $[N]$ such that $x(\{1, \ldots, i^*\})=1$ and $x(S)=0$ for any other $S \subseteq [N]$. To show recommending the agents to follow the social optimal strategy profile is persuasive, it suffices to show $x$ satisfies the constraints of the linear program \ref{eq:lp-private}.

For agent $i \leq i^*$, $x(S) = 0$ for all $S$ such that $i \notin 
S$ so the second constraint in \ref{eq:lp-private} is satisfied. Also,
\begin{align*}
\sum_{S: i \in S} x(S)(F(|S|)  -r(i)) =  x(\{1, \ldots, i^*\})(F(i^*) - r(i))
 =  F(i^*) - r(i)  
 \geq  F(i^*) - r(i^*) \geq 0,
\end{align*} 
hence the first constraint is also satisfied.

On the other hand, for agent $i > i^*$, $x(S) = 0$ for all $S$ such that $i \in S$ so the first constraint in \ref{eq:lp-private} is satisfied. Also,
\begin{align*}
& \sum_{S: i \notin S} x(S)(F(|S|+1)  -r(i)) \\
 = & F(i^*+1) - r(i) \\
\leq &  F(i^*+1) - r(i^*+1) \displaybreak[3]\\
\leq & \frac{r(i^*+1)}{\mu(1)} - r(i^*+1) \\
=  &\frac{\mu(0)}{\mu(1)}r(i^*+1)\\
 \leq &\frac{\mu(0)}{\mu(1)}r(i),
\end{align*}
so the second constraint is satisfied. Therefore $x$ is persuasive. 

The social welfare corresponding to $x$ is $\mu(1)\tilde{W}(i^*)$. Note that for any $S \subseteq [N]$, 
$$ |S|F(|S|) - \sum_{i \in S} r(i) \leq |S|F(|S|) - \sum_{i=1}^{|S|} r(i) = \tilde{W}(|S|) \leq \tilde{W}(i^*).$$
Therefore for any signaling scheme $\phi$, the social welfare with respect to $\phi$ is 
\begin{align*}
& \mu(0) \sum_{S \subseteq\{[N]\}} \phi(S | 0) \left(-\sum_{i \in S} r(i) \right)+  \mu(1)\sum_{S  \subseteq [N]} \phi(S | 1)\left(|S|F(|S|) - \sum_{i \in S} r(i)\right) \\
 \leq & \mu(1) \sum_{S  \subseteq [N]} \phi(S|1) \tilde{W}(i^*) \\ \leq & \mu(1)\tilde{W}(i^*),
 \end{align*}
hence $x$ gives the optimal signaling mechanism.

\end{proof}

\begin{proof}[\textbf{Proof of Proposition~\ref{thm:public-mixed-monotone}}]
$ $\newline

\noindent
Since agent $i$ is randomizing between moving and staying in equilibrium, assuming all other agents follow the given strategy profile, she should be indifferent in moving and staying, hence the expected payoff for agent $i$ choosing to move should be equal to $r(i)$. Let $N_{-ij}$ be the number of agents who choose to move besides agent $i$ and $j$, we have 
\begin{align*}
& q(p_j\expec[F(N_{-ij}+2)] + (1-p_j)\expec[F(N_{-ij}+1)]) = r(i) \\
\Rightarrow & p_j = \frac{\expec[F(N_{-ij}+1)]- r(i)/q}{\expec[F(N_{-ij}+1) - F(N_{-ij}+2)] }.
\end{align*}
Similarly, agent $j$ is indifferent in moving and staying given other agents' strategies, which gives 
\begin{equation}
p_i = \frac{\expec[F(N_{-ij}+1)]- r(j)/q}{\expec[F(N_{-ij}+1) - F(N_{-ij}+2)] }.
\end{equation}
Since $F$ is decreasing, $\expec[F(N_{-ij}+1) - F(N_{-ij}+2)] > 0$ and $r(i) \leq r(j)$ implies $p_i \leq p_j$.
\end{proof}

\begin{proof}[\textbf{Proof of Lemma~\ref{lem:threshold-equilibrium}}]

$ $\newline

\noindent
For any $t \in [\lihat(q), \uihat(q)] $, let $p = t + 1 - \lceil t \rceil$. We consider the utility for each agent $i$ choosing to move assuming all other agents follow their strategy in the threshold equilibrium.

For $i < \lceil t \rceil $, agent $i$' s expected utility for moving is 
\begin{align*}
q(pF( \lceil t \rceil ) + (1-p)F(\lceil t \rceil -1)) \geq qF( \lceil t \rceil )\geq F(\uihat(q)) \geq r(\uihat(q)) \geq r(i),
\end{align*}
hence she has no incentive to alter her strategy in the threshold equilibrium. 

For agent $\lceil t \rceil $, her expected utility for moving is $qF(\lceil t \rceil )$. If $t > \lihat(q)$, we have $\lihat(q) < \lceil t \rceil \leq \uihat(q)$, and by definition of $\lihat(q)$ and 
$\uihat(q)$, we have $qF(\lceil t \rceil ) = r(\lceil t \rceil)$ so agent $\lceil t \rceil $ would not alter her strategy. On the other hand, if $t =\lihat(q)$, by definition of $\lihat(q)$ we know agent $\lceil t \rceil $ would also choose to move, same as her strategy in the threshold equilibrium.

Finally, for agent $i > \lceil t \rceil $, her expected utility for moving is $q(pF( \lceil t \rceil +1) + (1-p)F( \lceil t \rceil ))$. In the case $t > \lihat(q)$, this utility is upper bounded by $qF( \lceil t \rceil ) \leq qF(\lihat(q)+1)$; and in the case $t = \lihat(q)$, we have $p=1$ and the utility for moving is also $qF(\lihat(q)+1)$. We have
$ qF(\lihat(q)+1) \leq r(\lihat(q)+1) \leq r(i)$ so agent $i$ would not alter her equilibrium strategy, which completes the proof for (i).

The ``only if'' part is straightforward and we omit the details. 
\end{proof}

\begin{proof}[\textbf{Proof of Theorem~\ref{thm:threshold-optimal}}]
\label{ap:appendix-a}

$ $\newline

\noindent
We first prove the following lemma.

\begin{lemma}
\label{lem:threshold-larger-than-optimal}
$W(q, n)$ is concave in $n$ and $W(q, \lihat(q)) \geq W(q, \lihat(q)+1)$.
\end{lemma}
\begin{proof}
Concavity is straightforward since $nF(n)$ is concave and $r(i)$'s are increasing in $i$ so $-\sum_{i=1}^n r(i)$ is also concave. To show $W(q, \lihat(q)) \geq W(q, \lihat(q)+1)$, we have
\begin{align*}
& W(q, \lihat(q)+1) - W(q, \lihat(q)) \\
=  &  q(\lihat(q)+1)F(\lihat(q)+1) - \sum_{i=1}^{\lihat(q)+1} r(i) - \left( q\lihat(q) F(\lihat(q)) - \sum_{i=1}^{\lihat(q)} r(i) \right) \\
=& q\lihat(q)(F(\lihat(q)+1) - F(\lihat(q))) + qF(\lihat(q)+1) - r(\lihat(q)+1) \\
\leq &  0,
\end{align*}
where the inequality is from $qF(\lihat(q)+1) \leq r(\lihat(q)+1)$ and $F$ is decreasing.
\end{proof}

Let $p=(p_1, \ldots, p_N)$ be an arbitrary equilibrium strategy profile under belief $q$. Define $\ione$, $\izero$ and $\imix$ as: $\ione \defeq \{1 \leq i \leq N: p_i=1\}$, $\izero \defeq \{1 \leq i \leq N: p_i=0\}$ and $\imix \defeq \{1 \leq i \leq N: 0 < p_i <1\}$, denoting the agents who moves, stays and randomizes between moving and staying. Recall $W(q, p)$ denote the principal's expected utility with respect to belief $q$ and  this strategy profile. We aim to show $W(q, p) \leq W(q, \lihat(q))$.

We have
\begin{align}
\label{eq:sender-utility-general}
W(q, p) = &\sum_{S \subseteq \imix} \prod_{j \in S} p_j \prod_{ j \in \imix \backslash S} (1-p_j) \left( q(|S| + |\ione|)F(|S| + |\ione|) - \sum_{j \in \ione \cup S} r(j)\right)\notag \\ 
\leq & \sum_{S \subseteq \imix} \prod_{j \in S} p_j \prod_{ j \in \imix \backslash S} (1-p_j) \left( q(|S| + |\ione|)F(|S| + |\ione|) - \sum_{j =1}^{|\ione| + |S|} r(j)\right) \displaybreak[3] \\
= & \sum_{n=0}^{|\imix|} \left( q(|\ione| + n)F(|\ione| + n) - \sum_{j=1}^{|\ione|+n}r(j) \right) \sum_{S \subseteq \imix: |S|=n} \prod_{j \in S} p_j \prod_{ j \in \imix \backslash S} (1-p_j)\notag \\
=& \sum_{n=0}^{|\imix|} W(q, |\ione| + n)  \sum_{S \subseteq \imix: |S|=n} \prod_{j \in S} p_j \prod_{ j \in \imix \backslash S} (1-p_j) . \notag
\end{align}

The inequality is from that $r(i)$'s are increasing in $i$. Let $X$ be the number of agents in $\imix$ that chooses to move. The right-hand side of the last equality of \eqref{eq:sender-utility-general} equals $\expec[W(q, |\ione|+X)]$. Since $W$ is concave, by Jensen's inequality we have 
\begin{equation}
\label{eq:jensen-uhat}
\expec[W(q, |\ione|+X)] \leq W(q, |\ione| + \expec[X]). 
\end{equation}
Therefore, $U(p_1, \ldots, p_N) \leq W(q, |\ione| + \expec[X])$. Next, we show $|\ione|  + \expec[X] \geq \lihat(q)$ in Lemma~\ref{lem:expec-move-larger-itilq}. 

\begin{lemma}
\label{lem:expec-move-larger-itilq}
 $|\ione|  + \expec[X] \geq \lihat(q)$.
\end{lemma}
\begin{proof}

Define $G(n) \defeq qF(n) , n \in [N]$. Since $G$ is convex, by Jensen's inequality, we have
\begin{equation}
\label{eq:g-convex}
G(|\ione| + \expec[X] + 1) \leq \expec[G(|\ione| + X+ 1)].
\end{equation}

For any agent $i \in \izero$, her expected utility if she chooses to move is $\expec[G(|\ione| + X+ 1)]$. Since she prefers staying, we have $\expec[G(|\ione| + X+ 1)] \leq r(i)$. On the other hand, for any agent $i \in \imix$, let $X^{(i)}$ denote the number of agents in $\imix$ besides agent $i$ who chooses to move. Agent $i$'s expected utility if she chooses to move is $\expec[G(|\ione| + X^{(i)} + 1)]$. Since agent $i$ is indifferent in moving and staying, we have $\expec[G(|\ione| + X^{(i)} + 1)] = r(i)$. Clearly $X^{(i)} \sdleq X$ where ``$\mathsf{sd}$'' denotes first-order stochastic dominance, and since $G$ is a decreasing function, we have $G(|\ione| + X+ 1) \sdleq G(|\ione| + X^{(i)} + 1)$. Therefore $\expec[G(|\ione| + X+ 1)] \leq \expec[G(|\ione| + X^{(i)} + 1)] = r(i)$. Summarizing the above, we have
\begin{equation}
\label{eq:expec-g-upper}
\expec[G(|\ione| + X+ 1)] \leq  \min_{i \in \izero \cup \imix} r(i).
\end{equation}
Let $j \defeq \min\izero \cup \imix$ and we show it must hold that $j \leq \uihat(q)$. Assume, for contradiction, $j > \uihat(q)$. First observe that since $j$ is the agent in $ \izero \cup \imix$ with the smallest index, it must hold $j \leq |\ione| + 1$, which implies $|\ione| \geq j-1 \geq \uihat(q)$. 
Meanwhile, if $j \in \izero$, we have $r(j) = \expec[G(|\ione| + X^{(j)} + 1)] \leq G(|\ione|+1)$; and if $j \in\imix$, we have $ r(j) \leq \expec[G(|\ione| + X + 1)] \leq G(|\ione|+1)$. Therefore, 
\begin{equation*}
r(\uihat(q)+1) \leq r(j) \leq   G(|\ione| + 1) \leq G(\uihat(q)+1) .
\end{equation*}
However by definition of $\uihat(q)$ it must be that $G(\uihat(q)+1) < r(\uihat(q)+1)$, which leads to a contradiction. 

From $j \leq \uihat(q)$, along with \eqref{eq:g-convex} and \eqref{eq:expec-g-upper}, we have 
\begin{equation*}
G(|\ione| + \expec[X] + 1) \leq \expec[G(|\ione| + X+ 1] \leq r(j) \leq r(\uihat(q)).
\end{equation*}
By definition of $\uihat(q)$ we have $G(\uihat(q)) \geq r(\uihat(q))$. In the case $G(\uihat(q)) > r(\uihat(q))$, we have $G(\uihat(q))  > G(|\ione| + \expec[X] + 1)$ hence $|\ione| + \expec[X] + 1> \uihat(q)$ since $G$ is decreasing. Note when $G(\uihat(q)) > r(\uihat(q))$, $\lihat(q) = \uihat(q)$, so we have  $|\ione| + \expec[X] \geq \lihat(q)$. On the other hand, if $G(\uihat(q)) = r(\uihat(q))$, we have $\lihat(q) < \uihat(q)$ so $|\ione| + \expec[X] + 1\geq \uihat(q) > \lihat(q)$ and we also get $|\ione| + \expec[X] \geq \lihat(q)$. 
\end{proof}

By Lemma~\ref{lem:threshold-larger-than-optimal}, $W(i)$ is decreasing on $i \geq \lihat(q)$. Since $|\ione|  + \expec[X] \geq \lihat(q)$ we have 
$W(q, |\ione| + \expec[X]) \leq W(q, \lihat(q))$. Along with \eqref{eq:sender-utility-general}, \eqref{eq:jensen-uhat} and Lemma~\ref{lem:expec-move-larger-itilq} we have
$$W(q,p) \leq \expec[W(q, |\ione| + X)] \leq W(q, \ione| + \expec[X]) \leq W(q, \lihat(q)), $$
completing the proof.

\end{proof}

\section{Upper Bound of \texorpdfstring{$\mu(1)$}{TEXT}}
\label{ap:mu1-upper}

In this section, we compute the upper bound $r(i^*+1)/F(i^*+1)$ for $\mu(1)$ given in Proposition~\ref{thm:private-condition}, which ensures recommending the agents following the social optimal strategy profile is persuasive, for several typical resource sharing functions and cost structures.

Specifically, we consider cost function $F(i) =1/i^\alpha$ for $\alpha=0.2, 0.4, 0.6, 0.8$, with different $\alpha$ controlling the curvature of $F$ and representing different resource sharing scenarios. We consider 3 different cost structures: constant costs , linear costs, and quadratic costs. We fix the total number of agents $N=20$ since $i^*$ does not depend on $N$. The result is given in Table \ref{table:table-mu1}.

\begin{table}[ht]
\begin{subtable}[h]{\textwidth}
\centering
\begin{tabular}{|c|c|c|c|c|c|c|c|c|c|c|}
  \hline
  \diagbox[width=3em]{$\alpha$}{$r$} & $0.1$ & $0.2$  & $0.3$  & $0.4$ & $0.5$ & $0.6$ & $0.7$ & $0.8$ & $0.9$ & $1.0$  \\
  \hline
 $0.2$ & $1$ & $1$  & $1$  & $1$ & $1$ & $1$ & $1$ & $1$ & $0.811$ & $0.808$  \\
  \hline
$0.4$ & $1$ & $1$  & $1$  & $0.621$ & $0.628$ & $0.653$ & $0.666$ & $0.696$ & $0.698$ & $0.776$  \\
  \hline
$0.6$  & $1$ & $0.422$  & $0.44$  & $0.459$ & $0.483$ & $0.58$ & $0.531$ & $0.606$ & $0.682$ & $0.758$  \\
  \hline
  $0.8$  & $0.237$ & $0.241$  & $0.261$  & $0.348$ & $0.435$ & $0.522$ & $0.609$ & $0.696$ & $0.783$ & $0.871$  \\
  \hline
\end{tabular}
\caption{Resource sharing function $F(i) = 1/i^\alpha$, with constant costs $r(i)=0.5r$.}
\end{subtable}

\vspace*{0.5cm}
\begin{subtable}[h]{\textwidth}
\centering
\begin{tabular}{|c|c|c|c|c|c|c|c|c|c|c|}
  \hline
  \diagbox[width=3em]{$\alpha$}{$r$} & $0.1$ & $0.2$  & $0.3$  & $0.4$ & $0.5$ & $0.6$ & $0.7$ & $0.8$ & $0.9$ & $1.0$  \\
  \hline
 $0.2$ & $1$ & $1$  & $0.836$  & $0.869$ & $0.888$ & $0.838$ & $0.849$ & $0.826$ & $0.93$ & $0.859$  \\
  \hline
$0.4$ & $0.617$ & $0.648$  & $0.65$  & $0.735$ & $0.762$ & $0.737$ & $0.666$ & $0.761$ & $0.857$ & $0.696$  \\
  \hline
$0.6$  & $0.464$ & $0.45$  & $0.527$  & $0.525$ & $0.459$ & $0.551$ & $0.643$ & $0.464$ & $0.521$ & $0.58$  \\
  \hline
  $0.8$  & $0.252$ & $0.243$  & $0.364$  & $0.289$ & $0.361$ & $0.433$ & $0.506$ & $0.279$ & $0.313$ & $0.348$  \\
  \hline
\end{tabular}
\caption{Resource sharing function $F(i) = 1/i^\alpha$, with linear  costs $r(i)=0.1ri$.}
\end{subtable}

\vspace*{0.5cm}
\begin{subtable}[h]{\textwidth}
\centering
\begin{tabular}{|c|c|c|c|c|c|c|c|c|c|c|}
  \hline
  \diagbox[width=3em]{$\alpha$}{$r$} & $0.1$ & $0.2$  & $0.3$  & $0.4$ & $0.5$ & $0.6$ & $0.7$ & $0.8$ & $0.9$ & $1.0$  \\
  \hline
 $0.2$ & $0.891$ & $0.947$  & $0.951$  & $1$ & $0.97$ & $0.868$ & $1$ & $0.824$ & $0.927$ & $1$  \\
  \hline
$0.4$ & $0.631$ & $0.78$  & $0.64$  & $0.854$ & $0.737$ & $0.885$ & $0.666$ & $0.761$ & $0.857$ & $0.952$  \\
  \hline
$0.6$  & $0.446$ & $0.63$  & $0.633$  & $0.525$ & $0.657$ & $0.441$ & $0.515$ & $0.588$ & $0.662$ & $0.735$  \\
  \hline
  $0.8$  & $0.302$ & $0.362$  & $0.291$  & $0.388$ & $0.485$ & $0.26$ & $0.303$ & $0.347$ & $0.39$ & $0.433$  \\
  \hline
\end{tabular}
\caption{Resource sharing function $F(i) = 1/i^\alpha$, with quadratic costs $r(i)=0.02ri^2$.}
\end{subtable}

 \caption{The upper bound $r(i^*+1)/F(i^*+1)$ of $\mu(1)$ for the private signaling mechanism that recommends every agent to stay when $\theta=0$ and recommends the first $i^*$ agents to move when $\theta=1$ to be optimal, given by Proposition~\ref{thm:private-condition}, for different resource sharing functions and cost structures. }\label{table:table-mu1}
\end{table}

\newpage

\section{Additional Computational Results}
\label{ap:numeric-mu}

In this section, we present computational results on how much social welfare can be generated by the optimal private and public signaling mechanism under different priors. We consider three different resource sharing function $F(i)=1/i^\alpha$ for $\alpha=0.2, 0.5, 0.9$; and three different cost structures: constant costs where $r(i)=0.5r$, linear costs where $r(i)=0.1ri$ and quadratic costs where $r(i)=0.02ri^2$, and $r$ is the cost coefficient. The results are given in Figure~\ref{fig:mu-numeric}.

\begin{figure}[h!]
     \centering
     \begin{subfigure}[b]{\textwidth}
         \centering
         \includegraphics[width=\textwidth]{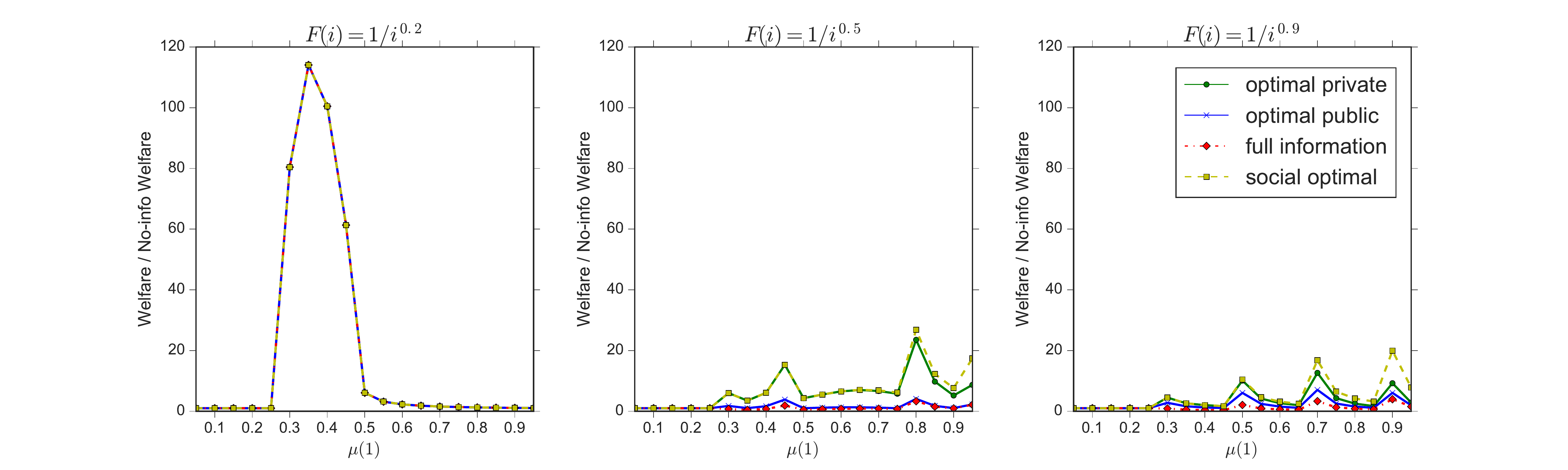}
         \caption{Constant costs: $r(i) = 0.5r$.}
         \label{fig:ratio-constant}
     \end{subfigure}
     \begin{subfigure}[b]{\textwidth}
         \centering
         \includegraphics[width=\textwidth]{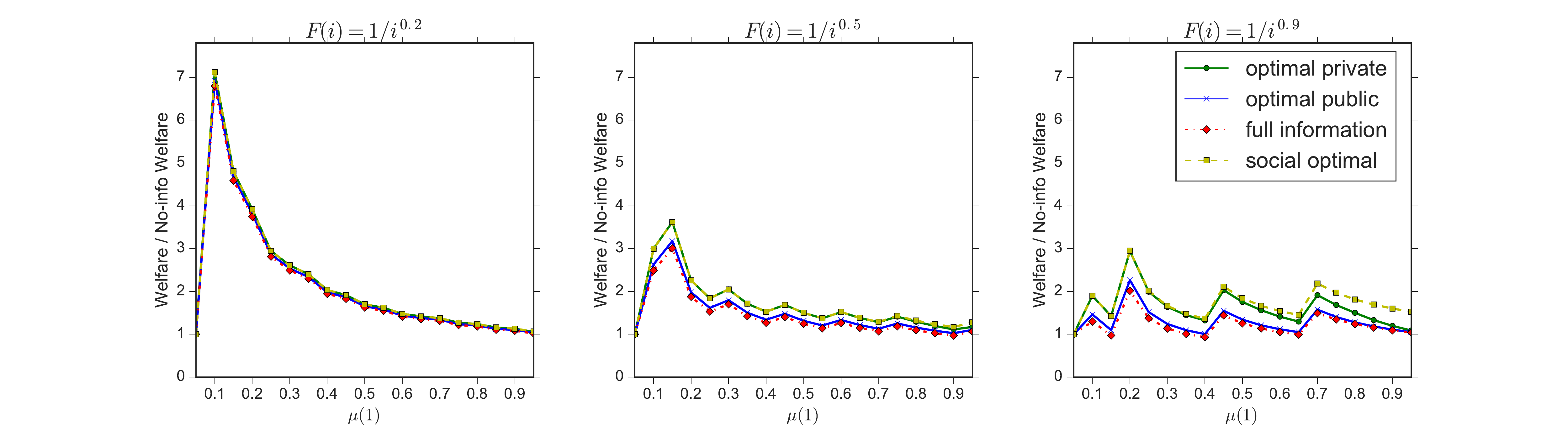}
         \caption{Linear costs: $r(i)=0.1ri$.}
         \label{fig:ratio-linear}
     \end{subfigure}
     \begin{subfigure}[b]{\textwidth}
         \centering
         \includegraphics[width=\textwidth]{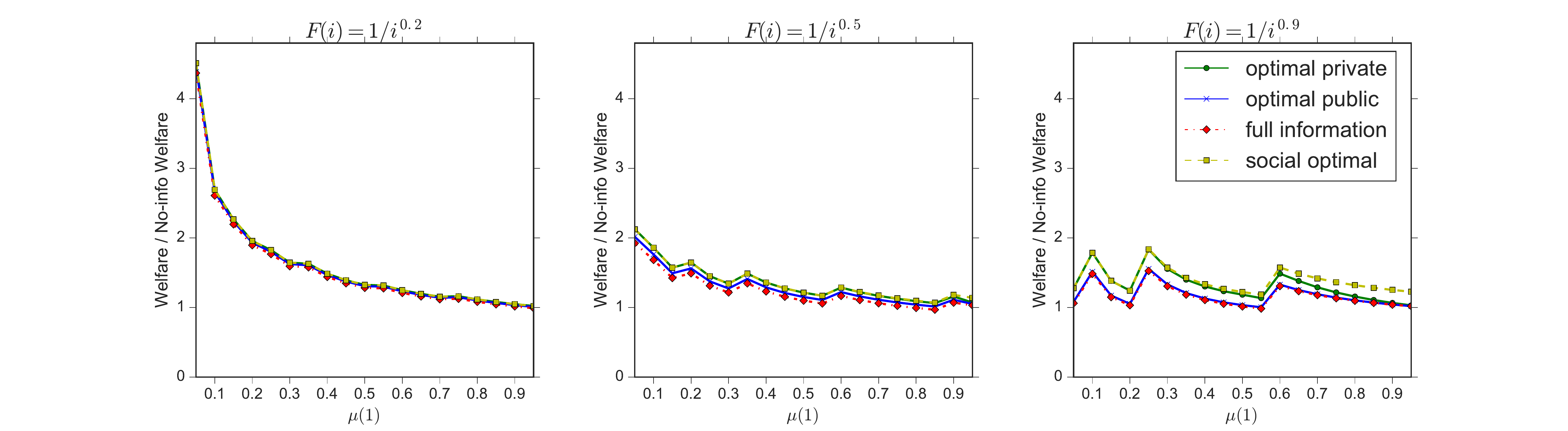}
         \caption{Quadratic costs: $r(i)=0.02ri^2$.}
         \label{fig:ratio-quad}
     \end{subfigure}
       \caption{Social welfare of the optimal signaling mechanisms and the benchmarks, under different prior beliefs, for different resource sharing and cost functions. In all experiments, $N=20$.}
  \label{fig:mu-numeric}
\end{figure}

\end{document}